\documentclass[journal]{IEEEtran}

\usepackage{cite}
\usepackage{algpseudocode}
\usepackage{graphicx}
\usepackage{subfigure}
\usepackage{amsthm}
\usepackage{amsmath}
\usepackage{amssymb}
\usepackage{algorithm}
\newtheorem{theorem}{Theorem}
\newtheorem{corollary}{Corollary}

\newtheorem{claim}{Claim}

\begin{document}
\title{ICE Buckets: Improved Counter Estimation for Network Measurement}
\author{Gil Einziger, Benny Fellman, Roy Friedman and Yaron Kassner
\thanks{Gil Einziger is with the department of Electrical Engineering, Politecnico di Torino, Italy. $<$gilga1983@gmail.com$>$.}%
\thanks{Benny Fellman is with the department of Electrical Engineering, Technion, Israel. $<$benny.fellman@gmail.com$>$.}%
\thanks{Roy Friedman and Yaron Kassner are with the department of computer science, Technion, Israel. $<$roy@cs.technion.ac.il$>$, $<$kassnery@gmail.com$>$ }
}

\date{}

\maketitle

\begin{abstract}

Measurement capabilities are essential for a variety of network applications, such as load balancing, routing, fairness and intrusion detection.
These capabilities require large counter arrays in order to monitor the traffic of all network flows. While commodity SRAM memories are capable of operating at line speed, they are too small to accommodate large counter arrays. Previous works suggested estimators, which trade precision for reduced space. However, in order to accurately estimate the largest counter, these methods compromise the accuracy of the smaller counters.
In this work, we present a closed form representation of the optimal estimation function.
We then introduce Independent Counter Estimation Buckets (ICE-Buckets), a novel algorithm that improves estimation accuracy
for all counters. This is achieved by separating the flows to buckets and configuring the optimal estimation function according
to each bucket's counter scale. We prove a tighter upper bound on the relative error and demonstrate an accuracy improvement of up to 57 times on real Internet packet traces.
\end{abstract}

\section{Introduction}
\subsection{Background}
\emph{Counter arrays} are essential in network measurements and accounting. Typically, measurement applications track several million flows~\cite{CounterArray1,CounterArray2}, and their counters are updated with the arrival of every packet. These capabilities are an important enabling factor for networking algorithms in many fields such as load balancing, routing, fairness, network caching and intrusion detection ~\cite{ApproximateFairness,IntrusionDetection,IntrusionDetection2,LoadBalancing,TinyLFU}.
Counter arrays are also used in popular approximate counting sketches such as \emph{multi stage filters}~\cite{CUSketch}
and \emph{count min sketch}~\cite{CountMinSketch},
as well as in network monitoring architectures~\cite{Ciuffoletti06architectureof,paxson1998architecture,BetterNetflow}.
Such architectures are used to collect and analyze statistics from many networking devices~\cite{HadoopArchitecture}.

Implementation of counter arrays is particularly challenging due to the requirement to operate at line speed.
Although commodity SRAM memories are fast enough for this task, they do not meet the space requirements of modern counter arrays.
Implementing a counter array entirely in SRAM is therefore very expensive~\cite{CounterArchitecture}.

Counter estimation algorithms use shorter counters, e.g., 12-bits instead of 32-bits, at the cost of a small error. 
Upon packet arrival, a counter is only incremented with a certain probability that depends on its current value.
In order to keep the relative error uniform, small values are incremented with high probability and large ones with low probability.
An \emph{estimation function} is used in order to determine these probabilities and estimate the true value of a counter.
Estimation functions can be \emph{scaled} to achieve higher counting capacity at the cost of a larger estimation error.

\subsection{Contributions}
\begin{figure}[t]
\subfigure[Twelve flows are estimated with a classic counter estimation array. Error for all counters is affected by the largest flow in the array (D).]{
\includegraphics[width=\columnwidth]
{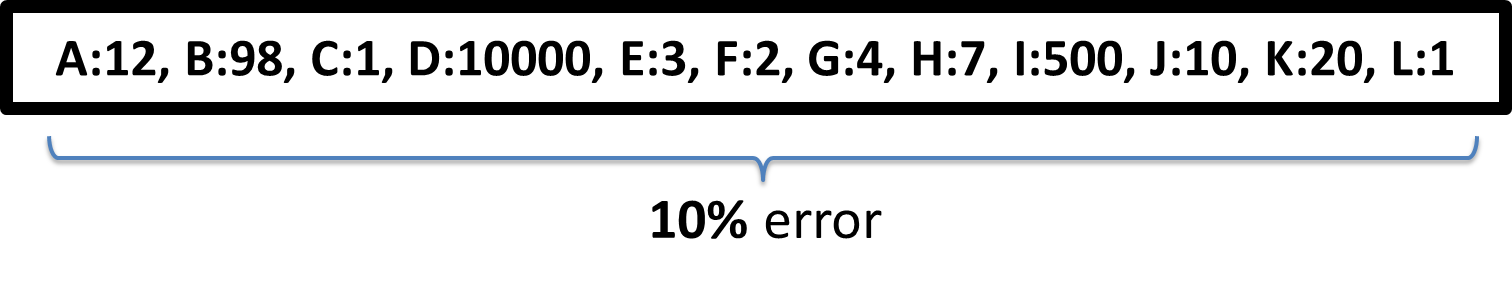}
\label{previousWork}
}
\subfigure[The same twelve flows are estimated with ICE-Buckets. In this example the flows are separated into four buckets. Error for each counter is affected only by the largest flow in its bucket. ]{
\includegraphics[width=\columnwidth]
{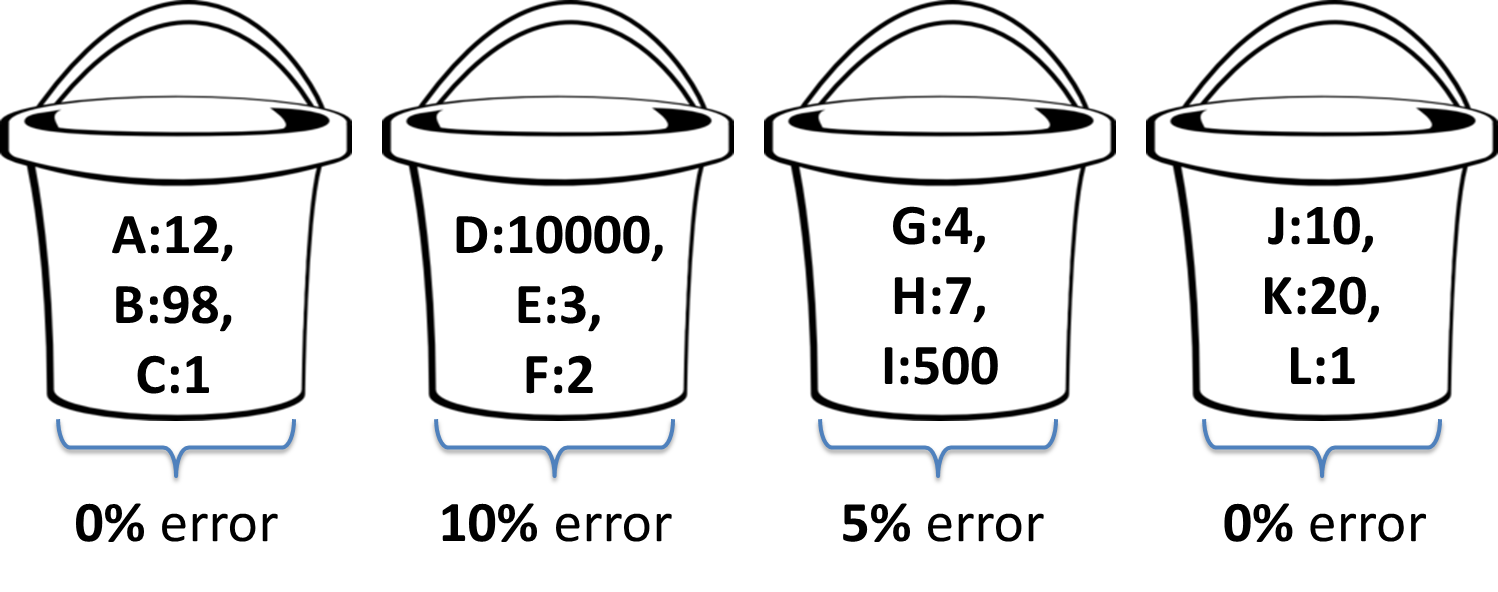}\label{ourwork}
}
\caption{An overview of ICE-Buckets vs. previous counter estimation approaches.}
\label{ICE vs previous}
\end{figure}

In this work we present \emph{Independent Counter Estimation Buckets (ICE-Buckets)}, a novel counter estimation technique that reduces the overall error by
efficiently utilizing multiple counter scales.

The main principle of ICE-Buckets is illustrated in Figure~\ref{ICE vs previous}.
In this example, the largest counter (D) can only be estimated with a large scale and a relative error of $10\%$. In the traditional approach, this error applies to all counters, as illustrated in Figure~\ref{previousWork}. Figure~\ref{ourwork} shows what happens when the array is partitioned into independent buckets. Counter D is still estimated with an error of $10\%$, but in this case the error applies only to counters within the same bucket. The other buckets are able to use smaller scales and enjoy lower relative error. Consequently, the overall error is reduced.

ICE-Buckets makes use of the optimal estimation function that was previously known only in recursive form. 
We present an explicit representation and provide an extended analysis for this function. 
We also present a rigorous mathematical analysis of ICE-Buckets that includes a very attractive upper bound for the overall relative error and a Chebyshev analysis to bound the probability that the error is above a given threshold.
We show that for traffic characteristics of real workloads this upper bound is up to 14 times smaller than that of previous works. 
Moreover, we show that the maximum relative error of ICE-Buckets is optimal.

We provide a lower and upper bounds for the space required to obtain a given counting capacity and error bound. 
We also analyze the error as a function of the maximal counting capacity and the number of estimation symbols.
This analysis provides us with the mathematical tools to configure ICE-Buckets parameters in an optimal manner.

We further show how to perform decrements and downscaling with ICE-buckets.
Yet, their complicated mathematical analysis is left for future work.

Additionally, we extensively evaluate ICE-Buckets with five real Internet packet traces and demonstrate an accuracy improvement of up to 57 times.
Finally, we show that ICE-Buckets can avoid global scale adjustments and still maintain similar accuracy. This configuration is more attractive for practical implementations.

In summary, we are the first to present a closed form explicit representation of an optimal estimation function.
This enables us to extensively study the various aspects of this function using rigorous mathematical analysis, including the relation between its relative error, memory complexity, estimation symbol range, and even bound the probability of the actual error exceeding a certain value.
We then propose the ICE-buckets technique, which divides counters into buckets, where each bucket is maintained with its own scale parameter, thereby greatly reducing the relative error.
ICE-Buckets is also analyzed, and we show a methodological way of configuring its parameters.
Finally, we simulate ICE-Buckets using 5 real-world traces and compare it to state of the art approaches, demonstrating its substantial benefits.

\subsection{Related Work}
While all counter arrays are required to monitor traffic at line speed, their implementations differ in the availability of monitored data. \emph{Offline} counter arrays can take as much as several hours to read from, while \emph{online} counter arrays can be read at line speed.
Naturally, offline counter arrays are used for high level tasks such as data analysis and identifying performance bottlenecks.
On the other hand, online counter arrays are used to answer low level queries such as what priority to give a certain flow, how much bandwidth it requires and where to route its packets. 

Hybrid DRAM/SRAM counter arrays~\cite{CounterArray2,CounterArray1} store only the least significant bits of each counter in SRAM and the rest of the counter in (slower) DRAM.
In \emph{CounterBraids}~\cite{CounterBraids}, counters are compressed in order to fit inside SRAM, but the decoding process is slow.
Alternatively, \emph{Randomized Counter Sharing (RCS)}~\cite{RCS} reduces the overhead required to maintain a flow to counter association.
In that solution, each flow is randomly associated with a large number of counters and on each packet arrival a random counter is incremented. Statistical methods are then used in order to decode flow values.
\emph{Counter Tree}~\cite{countertree} further reduces the memory requirements of RCS by introducing the concept of \emph{virtual counters}, each constructed from multiple physical counters organized in a tree structure such that large virtual counters span a path crossing multiple levels of the tree.
Here, each flow is associated to multiple virtual counters using a plurality of hash functions.
Hence, virtual counters share physical counters while flows share virtual counters and the virtual counters have variable size.
Alas, CounterBraids, RCS, as well as Counter Tree, are all offline due to their long complex decode time, while hybrid SRAM/DRAM architectures are offline since reading  requires accessing DRAM.
Interestingly, estimators like the one suggested in this paper, can further improve the space efficiency of RCS and Counter Tree at the expense of precision.

Brick~\cite{Brick} is an online counter array that encodes variable length counters. Brick can hold more counters as the average counter is shorter than the largest one. Unfortunately, the counting capacity is limited and the encoding becomes less efficient as the average counter value increases. Alternatively, sampling techniques~\cite{Sample1,BetterNetflow} and heavy hitters algorithms~\cite{HeavyHitters,SpaceSavings,BatchDecrement,Infocom2016} are able to monitor large flows. However, since they do not monitor all the flows, this type of solution is not always suitable.

Another popular approach for efficient flow statistics representation is shared counters.
In these schemes, there is no longer a guaranteed one to one correspondence between a counter and a flow.
Rather, some indirect hashing based mapping is maintained.
This enables eliminating maintaining flow identifiers and the respective associations.
Prominent example of these include \emph{count min sketch} (CMS)~\cite{CountMinSketch}, \emph{multi-stage filters}~\cite{HuffmanBF}, \emph{spectral Bloom filters} (SBF)~\cite{SpectralBloom} and their variants, as well as TinyTable~\cite{TinyTable}.
Shared counter techniques have a potential to work well with estimators, which can reduce the size of every shared counter.

Estimators are able to represent large values with small symbols at the price of a small error.  They can therefore be used to implement online counter arrays. This idea was first introduced by \emph{Approximate Counting}~\cite{ApproximateCounting} and was recently adapted to networking as \emph{Small Active Counters (SAC)}~\cite{SAC}. It was later improved by \emph{DISCO}~\cite{DISCOJournal} in order to provide better accuracy and support variable sized increments.

\cite{ANLSUpscaling} introduced a way to gradually increase the relative error as the counters grow.
\emph{CEDAR}~\cite{CEDAR} proved that their estimation function is optimal. \emph{CASE}~\cite{CASE} extended our analysis of the optimal estimation function to also include variable increments. They showed that the large flows can be tracked by a cache to improve accuracy. CASE can be deployed with any estimation technique including the one presented in this paper. 

In general, estimators require more space than sampling techniques, they provide accurate estimation for both small and large flows, and compared to Brick they enjoy significantly higher counting capacity at the price of a small relative error.

\subsection{Paper Organization}

The optimal estimation function is presented and analyzed in Section~\ref{optimal},
followed by the presentation and analysis of ICE-Buckets in Section~\ref{ICE-Buckets}. Section~\ref{Simulation Results} describes simulation results with real Internet packet traces. We conclude our work in Section~\ref{Conclusion}.
\section{Optimal Estimation Function}
\label{optimal}
\begin{table}[htp]
	{		
		\footnotesize
		\begin{center}
			\begin{tabular}{|c|c|}
				
				\hline
				Notation & Description\tabularnewline
				\hline
				\hline
				$M$ & Maximum possible number of packets.\tabularnewline
				\hline
				$L$ & Number of different possible symbols - a power of two.\tabularnewline
				\hline
				$\varepsilon_{max}$ & Maximal relative error (for a single counter). \tabularnewline
				\hline
				$\delta_{max}$ & Maximal coefficient of variation (CV) of the hitting time. \tabularnewline
					\hline
			
				$A_\epsilon(l)$ & The optimal estimation function with a scale of $\varepsilon = \varepsilon_{max}$
				\tabularnewline
				\hline
				
			\end{tabular}
			\normalsize
		\end{center}
		\caption{Notations}
		\label{notationdescription1}
	}
\end{table}
\subsection{Technical Background}
Consider the problem of counting up to $M$ packets, with a counter of only $\log_2L$ bits, where $\log_2L < \log_2\left(M+1\right)$ bits. We rely on an \emph{estimation function} $A:\{0,...,L-1\} \to \left[ 0,M \right]$, which accepts a symbol $l$ as input and returns an \emph{estimation value} for that symbol. $M$ is the required \emph{counting capacity} of the estimation function. For easy reference, the notations used in this section are summarized in Table~\ref{notationdescription1}

First, the symbol $l$ is initialized to zero. Upon arrival of a packet, we increment $l$ with probability $\frac{1}{D\left(l\right)}$, where $D\left( l\right) = A\left( l+1\right)-A\left( l\right)$. It is easy to verify that the expected estimation value of $l$ grows by one with each packet. Thus, the counter estimation is unbiased.

For example, for the estimation function $A(l)=2^l$, if at a certain point in time a symbol $l=3$ is used, its estimation value is $A(3)=8$. If another packet arrives at the flow, we increment the symbol with probability ${1 \over D(3)}={1 \over A(4)-A(3)}={1 \over 8}$. The estimation value is expected to change from $8$ to $16$ after 8 packet arrivals.

\subsection{Our Estimation Function}
We propose the following estimation function
\begin{equation} \label{optimal estimation function}
A_\epsilon(l)= \frac{(1+2\epsilon^2)^l-1}{2\epsilon^2}(1+\epsilon^2),
\end{equation}
where $\epsilon$ is a parameter of the algorithm. The motivation and benefits of this estimation function are discussed in Subsection~\ref{sec:analysis}.

\subsection{Upscale}
Upscale is a way to dynamically adjust the counter scale to the actual workload \cite{ANLSUpscaling}. It is useful in case the counting capacity~$M$ is unknown. We begin with a small $\epsilon$ that gives a counting capacity of $A_\epsilon(L-1)$ and dynamically increase it when necessary.
That is, when a symbol approaches $L-1$, we increase $\epsilon$ to $\epsilon'>\epsilon$. Then, we update all symbols to maintain unbiased estimation under the new scale.

Define $l'$ to be the largest integer such that $A_{\epsilon'}(l') \le A_{\epsilon}(l)$. For our estimation function, this value is
\begin{equation} \label{upscaling}
l'=\left\lfloor log_{1+2\epsilon'^2}\left(1+\frac{2\epsilon'^2A_{\epsilon}(l)}{1+\epsilon'^2}\right)\right\rfloor
\end{equation}

The correct estimation for symbol $l$ lies between $A_{\epsilon'}(l')$ and $A_{\epsilon'}(l'+1)$. We update to $l'+1$ with probability proportional to the difference between $A_{\epsilon}(l)$ and $A_{\epsilon'}(l')$:
$$\frac{A_{\epsilon}(l)-A_{\epsilon'}(l')}{A_{\epsilon'}(l'+1)-A_{\epsilon'}(l')}$$
and to $l'$ otherwise.  Algorithm~\ref{symbolupscale} describes the symbol upscale procedure.

\begin{algorithm}
\caption{Symbol Upscale}\label{symbolupscale}
\begin{algorithmic}[1]

\Procedure{SymbolUpscale}{$l$,$\epsilon$,$\epsilon'$}
\State $l'\gets
\left\lfloor log_{1+2\epsilon'^2}\left(1+\frac{2\epsilon'^2A_{\epsilon}(l)}{1+\epsilon'^2}\right)\right\rfloor $
\State $r\gets rand(0,1)$
\If {$r<\frac{A_{\epsilon}(l)-A_{\epsilon'}(l')}{A_{\epsilon'}(l'+1)-A_{\epsilon'}(l')}$}
	\State $l\gets l'+1$
\Else
	\State $l\gets l'$
\EndIf
\EndProcedure
\end{algorithmic}
\end{algorithm}
\subsection{Analysis}
\label{sec:analysis}
\subsubsection{Performance Metrics}
There are several metrics for the accuracy of an estimation function.
Throughout this paper, we discuss the quality of estimation mainly in terms of the \emph{root mean squared relative error (RMSRE)}, or \emph{relative error} in short. Denote $\hat{n}$ the random variable representing the estimation value of a flow after $n$ packets have arrived. The mean square relative error (MSRE) of a flow of size~$n$ is:
$$MSRE\left[n\right]=\mathbb{E}\left[\left(\frac{\hat{n}-n}{n}\right)^2\right]$$
and the root mean square relative error (RMSRE) is
$$RMSRE\left[n\right]=
\sqrt{\mathbb{E}\left[\left(\frac{\hat{n}-n}{n}\right)^2\right]}$$

We want the \emph{maximum relative error},
$$\epsilon_{max}=\max_{n\le M}RMSRE\left[n\right],$$
to be as small as possible.

When counting multiple flows, we can also measure the \emph{overall relative error}. Let $n_i$ be the true value of counter $i$. The overall relative error is the root mean square relative error over all $N$ counters,
$$\epsilon_{overall}=\sqrt{\frac{1}{N}\sum_{i}MSRE\left[n_i\right]}$$

Another metric for the accuracy of an estimation function is the \emph{hitting time}. The \emph{hitting time} is defined to be the random variable $T(l)$ that represents the amount of traffic required for a certain counter to be estimated as $A\left(l\right)$. The expected hitting time for symbol $l$ is simply $A\left(l\right)$ in our case, according to Theorem~2 in~\cite{CEDAR}, and the \emph{Coefficient of Variation (CV)} of the hitting time is $$CV\left[T\left(l\right)\right]=\frac{\sigma[T(l)]}{\mathbb{E}[T(l)]}.$$
Simplistically, $CV\left[T\left(l\right)\right]$ measures the relative error when the symbol becomes $l$.
In Theorem~\ref{optimality} below, we prove that our estimation function is optimal in terms of the \emph{maximum CV of the hitting time},
$$\delta_{max} = \max_{l<L}CV\left[T\left(l\right)\right].$$

\subsubsection{Optimality}
\label{sec:optimality}
An estimation function is considered optimal if it minimizes $\delta_{max}$ given $M$, the desired counting capacity.
To show that our estimation function is optimal, we rely on Theorems 3 and 4 in~\cite{CEDAR}, stating that a function that satisfies the following recursive formula is optimal with $\delta_{max}=\delta$.
\begin{align}
\forall{l}, A(l+1)-A(l)=&\frac{1+2\delta^2A(l)}{1-\delta^2} \label{recursive:1}\\
A\left(0\right)=&0 \label{recursive:2}
\end{align}
Furthermore, this function is unique.

\begin{theorem}\label{optimality}
$A_\epsilon(l)= \frac{(1+2\epsilon^2)^l-1}{2\epsilon^2}(1+\epsilon^2)$ is an optimal estimation function.
\end{theorem}
\begin{proof}
Clearly, $A_\epsilon(0)=0$ and thus condition \eqref{recursive:2} holds.

\begin{equation} \label{A recursion}
\begin{split}
A_{\epsilon}(l+1)
&= \frac{\left((1+2\epsilon^2)^{l+1}-1\right)
(1+\epsilon^2)}{2\epsilon^2}\\
&=\frac{\left((1+2\epsilon^2)(1+2\epsilon^2)^{l}
-(1+2\epsilon^2) + 2\epsilon^2\right)
(1+\epsilon^2)}{2\epsilon^2}\\
&=(1+2\epsilon^2)A_{\epsilon}\left(l\right)+(1+\epsilon^2)
\end{split}
\end{equation}
Therefore,
\begin{equation} \label{diff as A}
A_\epsilon(l+1)-A_\epsilon(l)=2\epsilon^2A_{\epsilon}\left(l\right)+(1+\epsilon^2)
\end{equation}
Choose
$$\delta = \sqrt{\frac{\epsilon^2}{1+\epsilon^2}}$$
to get
$$\epsilon^2 = \frac{\delta^2}{1-\delta^2}.$$
Putting this into \eqref{diff as A} we obtain
$$A_\epsilon(l+1)-A_\epsilon(l)
=2\frac{\delta^2A_{\epsilon}\left(l\right)}{1-\delta^2}+1+\frac{\delta^2}{1-\delta^2}=\frac{1+2\delta^2A_{\epsilon}(l)}{1-\delta^2}.$$
Hence, condition \eqref{recursive:1} holds.
Our estimation function satisfies conditions~\eqref{recursive:1} and~\eqref{recursive:2} and is therefore optimal.
\end{proof}

Thus, our estimation function is identical to the one given in CEDAR~\cite{CEDAR}, which was previously known only in recursive form and was only analyzed with respect to its hitting time.
In this work, we also analyze the relative error $RMSRE$, which is a natural metric to discuss.

\subsubsection{Relative Error}
\label{sec:relative error}
\begin{theorem} \label{CV theorem}
The optimal estimation function $(A_\epsilon)$ gives a relative error of
$$RMSRE\left[n\right]=\epsilon,\qquad\forall{n}$$
\end{theorem}
\begin{proof}
To prove this theorem, we use a technique similar to the one used in~\cite{ANLS}.\\
Let $Q_l(n)$ be the probability to have a symbol $l$ given that exactly $n$ packets have arrived at the flow. As mentioned before, the estimator is unbiased, thus
\begin{equation} \label{unbiased mean}
{\sum_{l}Q_l(n)A_{\epsilon}(l)=n.}
\end{equation}
In order to calculate $\epsilon$, we should first find the variance of the estimation value. We already know its mean, so let us find
$$\mathbb{E}\left[\hat{n}^2\right]=\sum_{l}Q_l(n)A_{\epsilon}^2(l).$$
We first compute
$\mathbb{E}\left[\hat{\left(n+1\right)}^2
-\hat{(n)}^2\right]$.
Recall that if the symbol is $l$, when a packet arrives the symbol is incremented with probability
$\frac{1}{D_\epsilon(l)} =\frac{1}{A_\epsilon(l+1)-A_\epsilon(l)}$ or remains unchanged with probability $1-\frac{1}{D_\epsilon(l)}$. Therefore,
\begin{align*}
&\mathbb{E}\left[\hat{\left(n+1\right)}^2
-\hat{(n)}^2\right]\\
&=\sum_{l}\left(A_\epsilon(l+1)^2-A_\epsilon(l)^2\right)\cdot \frac{1}{D_\epsilon(l)} Q_l(n)\\
&=\sum_{l}\left(A_\epsilon(l+1)+A_\epsilon(l)\right)Q_l(n)\\
\end{align*}
We can use \eqref{A recursion} to substitute $A_\epsilon(l+1)$ with $$(1+2\epsilon^2)A_{\epsilon}\left(l\right)+(1+\epsilon^2)$$ and obtain
$$\mathbb{E}\left[\hat{\left(n+1\right)}^2
-\hat{(n)}^2\right] =
\sum_{l}(1+\epsilon^2)(2A_\epsilon(l)+1)Q_l(n).$$
This can be separated to a constant multiplied by the unbiased mean, and another constant times the sum of a probability vector. We get
\begin{align*}
\mathbb{E}&\left[\hat{\left(n+1\right)}^2-\hat{(n)}^2\right]\\
&=(1+\epsilon^2)
\left(2\sum_{l}A_\epsilon(l)Q_l(n)+\sum_{l}Q_l(n)\right)\\
&=(1+\epsilon^2)(2n+1)\\
\end{align*}

Now, we can calculate $\mathbb{E}\left[\hat{n}^2\right]$
\begin{align*}
\mathbb{E}\left[\hat{n}^2\right]
&=\sum_{i=0}^{n-1}{\left(\mathbb{E}[\hat{(i+1)}^2]-\mathbb{E}[\hat{(i)}^2]\right)} +\mathbb{E}[\hat{0}^2]\\
&=\sum_{i=0}^{n-1}{(1+\epsilon^2)(2i+1)}
=(1+\epsilon^2)\frac{n\left(2n-1+1\right)}{2}\\
&=(1+\epsilon^2)n^2\\
\end{align*}
Hence, the variance is
$$\mathbb{V}\left[\hat{n}\right]
=\mathbb{E}\left[\hat{n}^2\right]
-\mathbb{E}\left[\hat{n}\right]^2
=\epsilon^2n^2$$
and the relative error is
$$RMSRE\left[n\right]
=\sqrt{\frac{\mathbb{V}\left[\hat{n}\right]}
{n^2}}
=\epsilon$$
\end{proof}
Note that the relative error is independent of $n$, and  therefore $\epsilon_{max}=\epsilon$. 
For comparison, DISCO's~\cite{DISCOJournal} estimation function,
$$DISCO(l) = \frac{(1+2\epsilon^2)^l-1}{2\epsilon^2},$$
also achieves a relative error bounded by $\epsilon$.
Thus, DISCO guarantees the same accuracy as the optimal estimation function but its counting capacity is $1+\epsilon^2$ times smaller. In most cases, this factor is negligible.

\subsubsection{Upscale Error}

We now show how to use linear programming to prove, for some $\epsilon$ and $\epsilon'$, that no upscale operation from $\epsilon$ to $\epsilon'$ increases the relative error to more than~$\epsilon'$.

Consider the change in variance when an upscale from~$\epsilon$ to~$\epsilon'$ occurs.
Let $Q_l$ denote the probability to have a symbol $l$ before upscale.
Recall that the probability to use $l'$ as defined in Algorithm~\ref{symbolupscale} is  $$\frac{A_{\epsilon'}(l'+1)-A_{\epsilon}(l)}{D_{\epsilon'}(l')}$$ and the probability to use $l'+1$ is $$\frac{A_{\epsilon}(l)-A_{\epsilon'}(l')}{D_{\epsilon'}(l')}.$$
The expected estimation value remains unchanged after upscale~\cite{CEDAR}.
Therefore, the change in variance is:
\begin{align*}
\Delta \mathbb{V}=\sum_{l}& Q_l
\frac{A_{\epsilon'}(l'+1)-A_{\epsilon}(l)}{D_{\epsilon'}(l')}
\cdot \left(A^2_{\epsilon'}(l')-A^2_\epsilon(l)\right)\\
&+Q_l\frac{A_{\epsilon}(l)-A_{\epsilon'}(l')}{D_{\epsilon'}(l')}
\cdot \left(A^2_{\epsilon'}(l'+1)-A^2_{\epsilon}(l)\right)\\
=\sum_{l}& Q_l
\frac{\left(A_{\epsilon'}(l'+1)-A_{\epsilon}(l)\right)
\left(A_{\epsilon}(l)-A_{\epsilon'}(l')\right)}
{D_{\epsilon'}(l')} \cdot\\
&\cdot \left(-A_{\epsilon'}(l')-A_\epsilon(l)
+A_{\epsilon'}(l'+1)+A_{\epsilon}(l)\right)\\
=\sum_{l}& Q_l\left(A_{\epsilon'}(l'+1)-A_{\epsilon}(l)\right)
\left(A_{\epsilon}(l)-A_{\epsilon'}(l')\right)\\
\end{align*}

If we find $\alpha\ge 0$ and $\beta\ge 0$ such that $\forall{0\le l<L}$,
\begin{equation}
\label{LP constaints}
\left(A_{\epsilon'}(l'+1)-A_{\epsilon}(l)\right)
\left(A_{\epsilon}(l)-A_{\epsilon'}(l')\right)\le \alpha A^2_{\epsilon}(l)+\beta A_{\epsilon}(l)
\end{equation}
we may conclude that the MSRE is
\begin{multline*}
MSRE[n]\le\frac{n^2\epsilon^2+\Delta \mathbb{V}}{n^2}\\
\le\frac{n^2\epsilon^2+\sum_{l} Q_l\left(\alpha A^2_{\epsilon}(l)+\beta A_{\epsilon}(l)\right)}{n^2}.
\end{multline*}

Recall that
$$\sum_{l} Q_l A^2_{\epsilon}(l)
=\mathbb{E}\left[A^2_\epsilon(l)\right]
=\mathbb{V}\left[\hat{n}\right]+\left(\mathbb{E}\left[\hat{n}\right]\right)^2
\le n^2(1+\epsilon^2).$$
Therefore, if \ref{LP constaints} holds for every $l$,
\begin{align*}
MSRE[n]&\le\frac{n^2\epsilon^2+\alpha n^2\left(1+\epsilon^2\right)+\beta n}{n^2}\\
&=\epsilon^2+\alpha \left(1+\epsilon^2\right)+\frac{\beta}{n}\\
&\le\epsilon^2+\alpha \left(1+\epsilon^2\right)+\beta.\\
\end{align*}

Define the following two variable linear problem: $$Minimize\qquad\epsilon^2+\alpha \left(1+\epsilon^2\right)+\beta$$
such that the constraint \eqref{LP constaints} holds for every $0\le l<L$.

We solved this simple LP for a wide range of parameters and found that the objective is minimized to $\epsilon'^2$ in all of these cases.
We therefore conjecture that the relative error is always bounded by $\epsilon'$ after an upscale operation from $\epsilon$~to~$\epsilon'$.

\subsubsection{Memory Complexity}
We now evaluate how many bits per symbol ($\lceil\log_2 L\rceil$) are needed to count to $M$ with a relative error of~$\epsilon$.
\begin{theorem}
\label{thm:memory} The required number of bits per symbol is 
\begin{multline*}
\lceil\log_2 L\rceil\\
=\left\lceil\log_2\left( 1+\frac{\ln\left(\left(2M+1\right)\epsilon^2+1\right)-\ln\left(1+\epsilon^2\right)}{\ln\left(1+2\epsilon^2\right)}\right)\right\rceil\\
=\log_2\ln \left((2M+1)\epsilon^2+1\right)+\log_2\epsilon^{-2}+\Theta(1)
\end{multline*}
\end{theorem}
\begin{proof}
We begin with the estimation value for symbol $L-1$, according to Equation~\eqref{optimal estimation function}:
$$M = \frac{\left(1+2\epsilon^2\right)^{L-1}-1}{2\epsilon^2}\cdot\left(1+\epsilon^2\right).$$
We solve for $L$:
$$\frac{2M\epsilon^2}{1+\epsilon^2}+1 = \left(1+2\epsilon^2\right)^{L-1}$$
\begin{multline*}
L = 1+\frac{\ln\left(\frac{2M\epsilon^2}{1+\epsilon^2}+1\right)}{\ln\left(1+2\epsilon^2\right)}\\
=1+\frac{\ln\left(\left(2M+1\right)\epsilon^2+1\right)-\ln\left(1+\epsilon^2\right)}{\ln\left(1+2\epsilon^2\right)}.\\
\end{multline*}
Thus,
\begin{multline*}
\lceil\log_2 L\rceil\\
=\left\lceil\log_2\left( 1+\frac{\ln\left(\left(2M+1\right)\epsilon^2+1\right)-\ln\left(1+\epsilon^2\right)}{\ln\left(1+2\epsilon^2\right)}\right)\right\rceil.
\end{multline*}

Next, we bound $L$ from both directions to get a simpler expression.
We start from below. We use the inequality
\begin{equation}
\ln (1+x) < x\text{ for }x>0 \label{eq:lnl}
\end{equation}
to obtain
\begin{multline}
L \ge 1 + \frac{\ln \left((2M+1)\epsilon^2+1\right)-\epsilon^2}{2\epsilon^2}\\
=\frac{1}{2} + \frac{\ln \left((2M+1)\epsilon^2+1\right)}{2\epsilon^2}. \label{eq:memory lower bound}
\end{multline}

Next, we bound $L$ from above. We use the inequality
\begin{equation}
\ln (1+x) > \frac{x}{1+x} \label{eq:lng}
\end{equation} for $x>0$,
\begin{multline}
L\le 1+\frac{\ln\left((2M+1)\epsilon^2+1\right)-\frac{\epsilon^2}{1+\epsilon^2}}{\frac{2\epsilon^2}{1+2\epsilon^2}}\\
\le 1+\frac{1+2\epsilon^2}{2\epsilon^2}\ln\left((2M+1)\epsilon^2+1\right)
-\frac{1+2\epsilon^2}{2+2\epsilon^2}
~\label{eq:tightMiddlestep}
\end{multline}
Notice that $\frac{1+2\epsilon^2}{2+2\epsilon^2} > \frac{1}{2}$, therefore we continue by replacing $\frac{1+2\epsilon^2}{2+2\epsilon^2}$ with $\frac{1}{2}$ in Inequality~\eqref{eq:tightMiddlestep}. We get: 
\begin{multline}
L<\frac{1}{2}+\left(1+2\epsilon^2\right)\frac{\ln\left((2M+1)\epsilon^2+1\right)}{2\epsilon^2}.
\label{eq:tight memory upper bound}
\end{multline}
Since $L\ge 2$, we get from Inequality~\eqref{eq:tight memory upper bound} that $$1.5 \le \left(1+2\epsilon^2\right)\frac{\ln\left((2M+1)\epsilon^2+1\right)}{2\epsilon^2}.$$
Thus, replacing the number $\frac{1}{2}$ from Inequality~\eqref{eq:tight memory upper bound} with the above expression divided by 3, we
obtain $$L\le \frac{4\left(1+2\epsilon^2\right)}{6}\cdot\frac{\ln\left((2M+1)\epsilon^2+1\right)}{\epsilon^2}.$$
$\epsilon<1$ and therefore $$L<2\cdot \frac{\ln\left((2M+1)\epsilon^2+1\right)}{\epsilon^2}.$$

Consequently, the amount of bits required to guarantee an error of at most $\epsilon$ and a counting capacity of at least $M$ is
\begin{equation}
\left\lceil \log_2 L\right\rceil
\le \left\lceil 1+\log_2\ln\left((2M+1)\epsilon^2+1\right)+\log_2\epsilon^{-2}\right\rceil
\label{eq:final memory lower bound}
\end{equation}
and on the other hand (from Inequality~\eqref{eq:memory lower bound})
\begin{equation}
\left\lceil \log_2 L\right\rceil
\ge \left\lceil -1+\log_2\ln\left((2M+1)\epsilon^2+1\right)+\log_2\epsilon^{-2}\right\rceil.
\label{eq:final memory upper bound}
\end{equation}
Therefore, the amount of memory required is
\begin{equation*}
\left\lceil\log_2 L\right\rceil = \log_2\ln \left((2M+1)\epsilon^2+1\right)+\log_2\epsilon^{-2}+\Theta(1).
\end{equation*}
\end{proof}

\begin{corollary} For a constant $\epsilon$, the required number of bits is $O(\log\log M)$. \end{corollary}

\begin{figure}[ht]
\centering
\includegraphics[width=\linewidth]{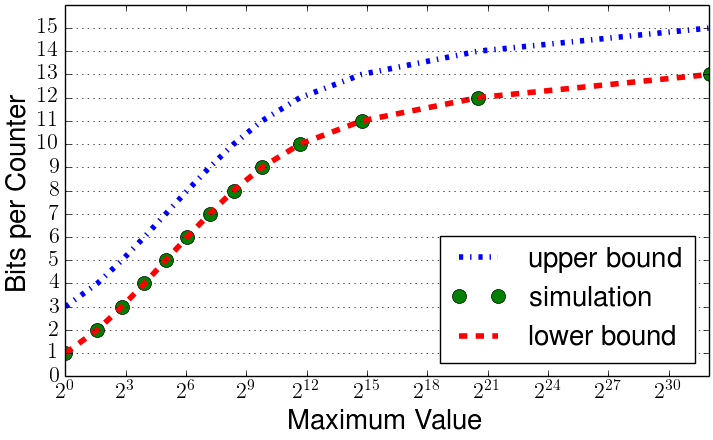}
\caption{Simulation and bounds for the number of bits per symbol required to count with different values of $M$ for $\epsilon=2^{-5}$.}
\label{fig:memory_bounds}
\end{figure}

In figure~\ref{fig:memory_bounds}, we demonstrate the accuracy of the memory bounds we derive in Theorem~\ref{thm:memory} (inequalities~\eqref{eq:final memory lower bound} and \eqref{eq:final memory upper bound}). We simulate different numbers of bits per symbol with $\epsilon=2^{-5}$. For each value of $L$ we then calculate the counting capacity $M$ and from it the memory bounds from inequalities~\eqref{eq:final memory lower bound} and \eqref{eq:final memory upper bound}. For these parameters, we can see that the lower bound is tight, and the upper bound is only two bits larger. In addition, the figure shows that with an approximation of~$\epsilon=2^{-5}$, 13 bits are sufficient to count up to more than~$2^{32}$.

\subsubsection{$\epsilon$ as a function of $M$ and $L$}
Computing the smallest $\epsilon$ sufficient to represent a counter $M$ with $L$ possible symbols may be useful for several purposes. First, if we know $M$ in advance and we have limited memory, we can detect the optimal parameter to use and avoid the upscale phase. Second, it enables theoretical comparison of the algorithm's relative error.

To achieve a counting capacity $M$ with $\log_2 L$ bits, the error should be the $\epsilon$ that solves
\begin{equation}
M=\frac{\left(1+2\epsilon^2\right)^{L-1}-1}{2\epsilon^2}\cdot\left(1+\epsilon^2\right).
\label{eq:M}
\end{equation}
$$\frac{2M\epsilon^2}{1+\epsilon^2}
= \left(1+2\epsilon^2\right)^{L-1}-1.$$

We show two different $\epsilon$'s and show that one gives a value greater than $M$ when put into Equation~\eqref{eq:M} and the other gives a value smaller than $M$.

\begin{claim}
$$\epsilon^2\le3\frac{\ln{6M\over L-1}}{L-1}.$$
\end{claim}
\begin{proof}
First, let
$$\epsilon^2=3\frac{\ln{6M\over L-1}}{L-1}.$$
By developing Equation~\eqref{eq:M}, we get
\begin{multline*}
A_\epsilon(L-1)=\frac{\left(1+2\epsilon^2\right)^{L-1}-1}{2\epsilon^2}\cdot \left(1+\epsilon^2\right)\\
= \frac{e^{\ln\left(1+2\epsilon^2\right)\cdot(L-1)}-1}{2\epsilon^2}\cdot \left(1+\epsilon^2\right).
\end{multline*}
Now we use Inequality~\eqref{eq:lng} to obtain
$$A_\epsilon(L-1)>\frac{e^{\frac{2\epsilon^2}{1+2\epsilon^2}\cdot(L-1)}-1}{2\epsilon^2}.$$
Since $\epsilon<1$, we have 	
\begin{multline*}
A_\epsilon(L-1)>\frac{e^{\frac{2\epsilon^2}{3}\cdot(L-1)}-1}{2\epsilon^2}
=\frac{e^{2\ln\frac{6M}{L-1}}-1}{6\frac{\ln{\frac{6M}{L-1}}}{L-1}}\\
=\frac{\left(\left(\frac{6M}{L-1}\right)^2-1\right)(L-1)}{6\ln\frac{6M}{L-1}}.
\end{multline*}
Now, we can use Inequality~\eqref{eq:lnl} to obtain
$$A_\epsilon(L-1)>\frac{\left(\frac{6M}{L-1}-1\right)\left(\frac{6M}{L-1}+1\right)(L-1)}{6\left(\frac{6M}{L-1}-1\right)}>M.$$
Since $M$ is increasing with $\epsilon$ and the chosen $\epsilon$ gives a counting capacity greater than $M$, to achieve a counting capacity of exactly $M$, $\epsilon$ must be less than or equal to $3\frac{\ln{6M\over L-1}}{L-1}$.
\end{proof}

\begin{claim}
$$\epsilon^2\ge\frac{\ln{2M+1\over 2L-1}}{2(L-1)}.$$
\end{claim}
\begin{proof}
Let
$$\epsilon^2=\frac{\ln{2M+1\over 2L-1}}{2(L-1)}.$$
As before,
\begin{multline*}
A_\epsilon(L-1)
=\frac{\left(1+2\epsilon^2\right)^{L-1}-1}{2\epsilon^2}\cdot \left(1+\epsilon^2\right)\\
= \frac{e^{\ln\left(1+2\epsilon^2\right)\cdot(L-1)}-1}{2\epsilon^2}\cdot \left(1+\epsilon^2\right).
\end{multline*}
We use Inequality~\eqref{eq:lnl} to get
\begin{multline*}
A_\epsilon(L-1)
\le \frac{e^{2\epsilon^2\cdot(L-1)}-1}{2\epsilon^2}\cdot \left(1+\epsilon^2\right)\\
=\frac{e^{\ln{2M+1\over 2L-1}}-1}{2\frac{\ln{2M+1\over 2L-1}}{2(L-1)}}\cdot \left(1+\frac{\ln{2M+1\over 2L-1}}{2(L-1)}\right)\\
=\frac{{2M+1\over 2L-1}-1}{2}\cdot
\left(1+\frac{2(L-1)}{\ln{2M+1\over 2L-1}}\right)
\end{multline*}
Returning to Inequality~\eqref{eq:lng}, we can get
\begin{multline*}
A_\epsilon(L-1)
\le\frac{{2M+1\over 2L-1}-1}{2}\cdot
\left(1+\frac{2(L-1)\cdot {2M+1\over 2L-1}}{{2M+1\over 2L-1}-1}\right)\\
=\frac{{2M+1\over 2L-1}-1+2(L-1)\cdot {2M+1\over 2L-1}}{2}
=M.
\end{multline*}
We have shown that an $\epsilon$ of $\frac{\ln{2M+1\over 2L-1}}{2(L-1)}$ gives a counting capacity of at most $M$. Since $M$ is increasing with $\epsilon$, we conclude that $\epsilon$ must be greater or equal to $\frac{\ln{2M+1\over 2L-1}}{2(L-1)}$.
\end{proof}
Using Bolzano-Weierstrass theorem, we conclude that there exists an $\epsilon$ such that
$$\frac{\ln{2M+1\over 2L-1}}{2(L-1)}\le\epsilon^2\le 3\frac{\ln{6M\over L-1}}{L-1}$$
for which Equation~\eqref{eq:M} holds.

The two sides of the inequality may be asymptotically distinct. But, for interesting cases, this does not happen.
If $M\le 2L-1$, then it is sufficient to use one more bit than $\lceil\log_2 L\rceil$ to represent $M$ exactly. Therefore, the interesting case is when $M>2L-1$. In this case, it is simple to see that both sides of the above equation are asymptotically $\Theta\left(\frac{\ln\frac{M}{L}}{L}\right)$. Hence, for $M>2L-1$,
$$\epsilon^2=\Theta\left(\frac{\ln\frac{M}{L}}{L}\right).$$

\begin{figure}[htp]
\centering
\includegraphics[width=\linewidth]{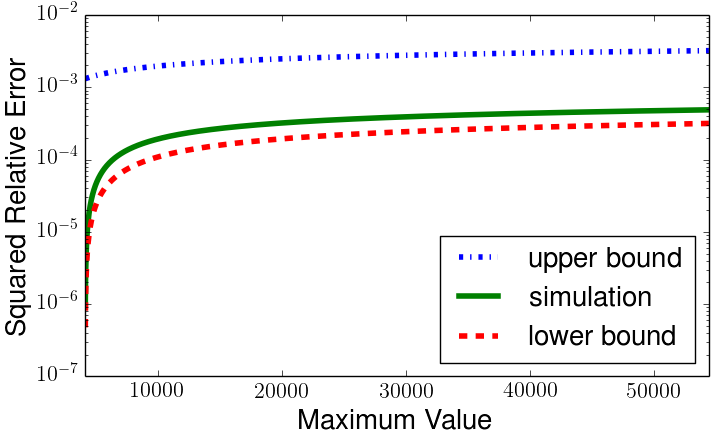}
\caption{Squared relative error and bounds with a 12-bit symbol for different values of $M$.}
\label{fig:error_bounds}
\end{figure}

Figure~\ref{fig:error_bounds} demonstrates the findings. To create it, we simulated $M$ with $L=4096$ and errors ranging between $2^{-20}$ and $2^{-11}$. For each $M$, we then calculated the lower and upper bound. We first see that the true error indeed lies between the lower and the upper bound. In this case, the lower bound is much tighter than the upper bound. For small values of $M$, where $M<2L-1=9191$, the upper bound is inaccurate. However, for larger values of $M$, we see that the two bounds and the true error are asymptotically the same, as shown above.

\subsubsection{Chebyshev Analysis}
The RMSRE metric may be unsuitable for some applications because it only describes the maximum \emph{expected} relative error rather than the maximum relative error. Some applications may require determining with certainty $1-\rho$ that the relative error $\frac{|\hat{n}-n|}{n}$ is no more than $\beta$. We can obtain that guarantee using Chebishev's inequality, according to which
$$Pr\left(\frac{|\hat{n}-n|}{n}\ge k\epsilon\right)
=Pr\left(|\hat{n}-n|\ge k\sigma\right)\le \frac{1}{k^2},$$
where $\sigma=n\epsilon$ is the standard deviation. Choosing $k=\frac{\beta}{\epsilon}$ gives us a probability $\rho = \frac{\epsilon^2}{\beta^2}$. This allows us to choose the parameter $\epsilon$ according to any given $\beta$ and $\rho$. For example, to achieve a relative error of more than $\beta=10\%$ with probability no more than $\rho=1\%$, we can use $\epsilon=\sqrt{\beta^2\rho}=1\%$.

\subsection{Decrementing Counters}
Decrementing counters can be useful for applications that need to ``forget'' old, less relevant, values, e.g., when we wish to count flows in a sliding window~\cite{Infocom2016}.
The process of decrementing a flow is very similar to the process of incrementing a flow.
Instead of incrementing $l$ with probability $\frac{1}{D(l)}$, we decrement it with probability $1\over D(l-1)$.
\begin{theorem}
For an unbounded counter, the estimation is unbiased after any number of increments or decrements.
\end{theorem}
\begin{proof}
Let $l_t$ be the random variable that represents the counter at time $t$. Let $n_t$ be the number of increments minus the number of decrements until and including time $t$. Assume by induction on $t$ that the estimation is unbiased at time $t$, i.e. $\mathbb{E}\left[A(l_t)\right] = n_t$. We next show that the estimation remains unbiased at time $t+1$. If the symbol is incremented at time $t+1$,
\begin{multline*}
\mathbb{E}\left[A(l_{t+1)}\right]
= \sum_{x}{A(l_{t+1})Pr\left(l_t=x\right)}\\
= \sum_{x}{\left(\frac{1}{D(x)}A(x+1)+\left(1-\frac{1}{D(x)}\right)A(x)\right)Pr\left(l_t=x\right)}\\
= \sum_{x}{\frac{1}{D(x)}\left(A(x+1)-A(x)\right)Pr\left(l_t=x\right)}\\
+\sum_{x}{A(x)Pr\left(l_t=x\right)}\\
= \sum_{x}{Pr\left(l_t=x\right)}+n_t = 1+n_t
\end{multline*}
Similarly, if the symbol is decremented at time $t+1$,
\begin{multline*}
\mathbb{E}\left[A(l_{t+1)}\right]
= \sum_{x}{A(l_{t+1})Pr\left(l_t=x\right)}\\
= \sum_{x}{\frac{1}{D(x-1)}A(x-1)Pr\left(l_t=x\right)}\\
+\sum_{x}{\left(1-\frac{1}{D(x-1)}\right)A(x)Pr\left(l_t=x\right)}\\
= \sum_{x}{\frac{1}{D(x-1)}\left(A(x-1)-A(x)\right)Pr\left(l_t=x\right)}\\
+\sum_{x}{A(x)Pr\left(l_t=x\right)}\\
= \sum_{x}{-Pr\left(l_t=x\right)}+n_t = -1+n_t
\end{multline*}
Therefore, the estimation is unbiased at time $t+1$. At time 0 the estimation is 0 and therefore unbiased. In conclusion, the estimation is unbiased at every time $t$.
\end{proof}

\subsection{Downscale}
In scenarios where counters are decremented, we may find ourselves in a situation where counters are small yet their scale is large, resulting in a quickly increasing error. In this case, we may want to \emph{downscale} the counters.

Let $\epsilon$ be the current error parameter and $\epsilon'$ a smaller, desired error. We downscale all counters after checking that they can all be represented with $A_{\epsilon'}$. The process of checking whether all counters can be represented with $A_{\epsilon'}$ is repeated until all counters are small enough. In this process, we iterate over every symbol $l$ and check that $A_\epsilon(l)<A_{\epsilon'}(L-1)$. If this iteration and downscaling of the counters takes time and $U$ updates are performed during that time, we might want to check instead that $A_\epsilon(l+U)<A_{\epsilon'}(L-1)$. This guarantees that once all counters are downscaled, they can still be represented with $A_{\epsilon'}$.

Then, we update all symbols to maintain unbiased estimation under the new scale. This is done with the Symbol Downscale procedure, which is identical to the Symbol Upscale procedure defined in Algorithm~\ref{symbolupscale}. The only difference is that this time $\epsilon'<\epsilon$.

However, changing the parameter of the estimation function to $\epsilon'$ does not reduce the relative error to $\epsilon'$. It is important to distinct between the error parameter and the actual error.

Reducing the scale of a counter complicates the analysis. Therefore, in this work we assume that no decrements and no downscale occur, i.e., packets only arrive and do not leave.

\section{ICE-Buckets}
\label{ICE-Buckets}
\begin{table*}[htp]
	{
		\footnotesize
		\begin{center}
		\begin{tabular}{|c|c|}
			\hline
			Notation & Description\tabularnewline
			\hline
			\hline
			$M$ & Maximum possible number of packets.\tabularnewline
			\hline
			$L$ & Number of different possible symbols - a power of two.\tabularnewline
			\hline

			$S$ & Number of symbols per bucket.\tabularnewline
			\hline
			$N$ & Number of flows.\tabularnewline
			\hline

			$B$ & Number of buckets $\left(\frac{N}{S}\right)$. For simplicity assume that $N$ is a multitude of $S$.\tabularnewline
			\hline

			$F_{ij}$ & Symbol array of size $B\times S$. $i$ runs over the buckets and $j$ runs over the flows in each bucket. Each symbol is $\log_2L$ bits wide.\tabularnewline
			\hline
			$w_i$ & Scale parameter for bucket $i$.\tabularnewline
			\hline
			$\epsilon_{step}$ &  The difference between consecutive estimation errors.\tabularnewline
			\hline
			$\epsilon_w$ &  Estimation error for bucket with scale parameter $w$ ($\epsilon_{step}\cdot w$).\tabularnewline
			\hline
			$\varepsilon_{overall}$ & Overall relative error (averaged on all counters). \tabularnewline
			\hline
			$E$ &  Number of different possible estimation scales - a power of two.\tabularnewline
			\hline
			$T$ &  Space allocated for the data structure (bits).\tabularnewline
			\hline
		\end{tabular}	
		\normalsize
		\end{center}
		\caption{Notations}
		\label{notationdescription}
		}
\end{table*}

\subsection{Overview}
We now describe ICE-Buckets, a data structure that uses the optimal estimation function with a different scale for each bucket. First, notations specific to this section are given in Table~\ref{notationdescription}.
ICE-Buckets uses a base error parameter that is called $\epsilon_{step}$. Symbols are separated into buckets and each bucket maintains a scale parameter $w_i$ of size $\log_2 E$ bits, where $E$ is a parameter of the data structure. To estimate counters in a bucket with scale~$w$, we use the estimation function $A_{\epsilon_{step}\cdot w}$. We pay special attention to the additional memory that is allocated for each bucket and make sure that this overhead is small. Figure~\ref{ICE architecture} demonstrates this basic architecture.
\begin{figure}[h!]
\centering
\includegraphics[width=85mm]{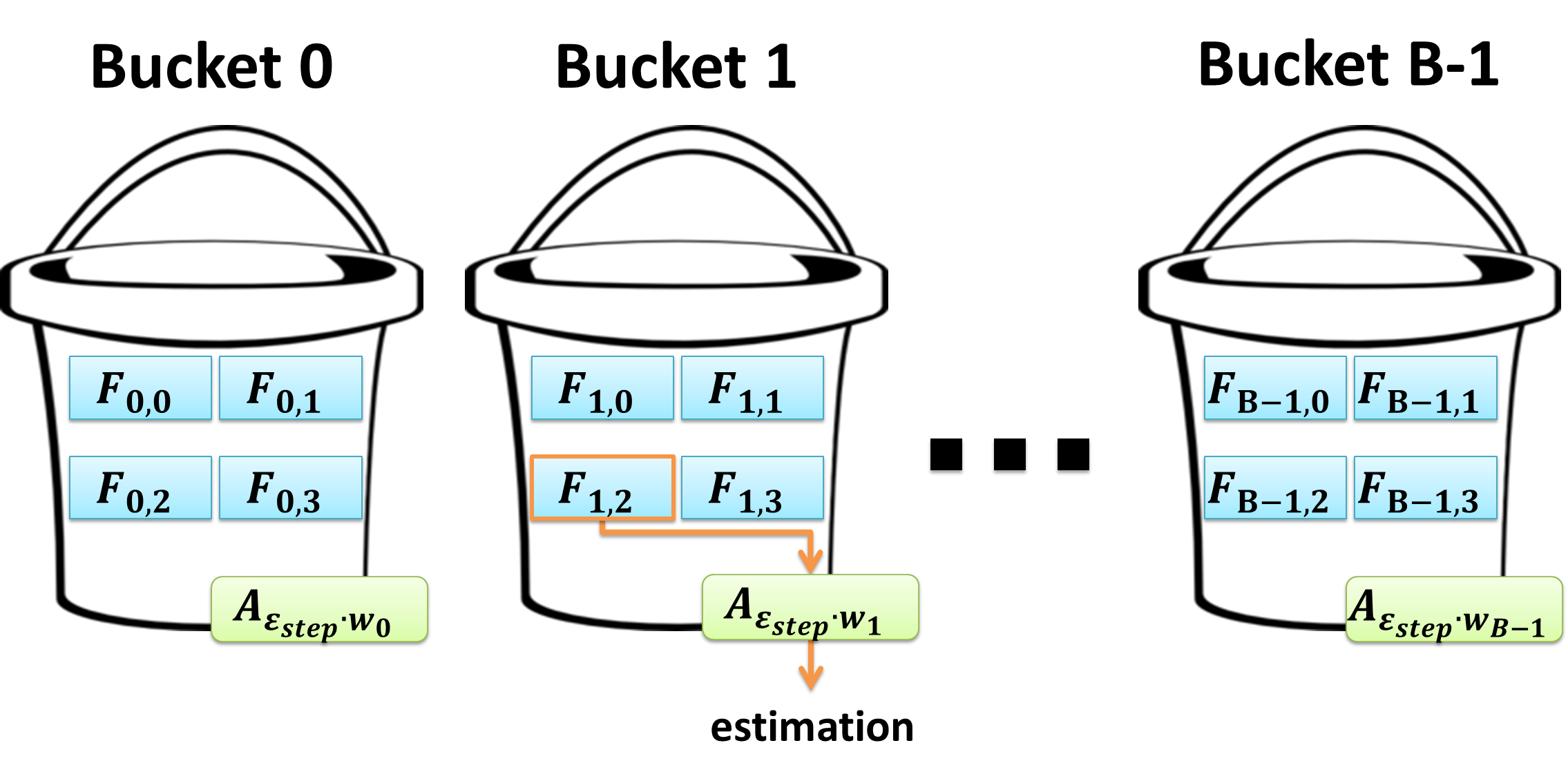}
\caption{An ICE-Buckets data structure with four flows per bucket. Bucket $i$ has a scale parameter $w_i$, which is used by the estimation function $A_{\epsilon_{step}\cdot w_i}$ to decode the symbols ($F_{i,j}$). In this example, to estimate counter 2 in bucket 1, $A_{\epsilon_{step}\cdot w_1}(F_{1,2})$ is computed.}
\label{ICE architecture}
\end{figure}
\subsection{Algorithm}
Initially, all symbols are set to zero. Since each bucket's scale is different, we first associate each flow-id with a bucket.
To estimate the value of flow $f$, we use $A_{\epsilon_{w_i}}\left(F_{ij}\right)$ where
$i=\left\lfloor\frac{f}{S}\right\rfloor$
and
$j=(f\mod{S}).$

The estimation values for bucket $i$ are calculated with an optimal estimation function of scale $\epsilon_{w_i}=\epsilon_{step}\cdot w_i$.
An optimal function with  $\epsilon = 0$ is defined as the identity function, i.e.,
$\forall{l}: A_{0}\left(l\right)=l.$

The increment process is straightforward. When a packet arrives, first we find the associated bucket $i$ and the index in the bucket, $j$. Then, we increment $F_{ij}$ with probability $\frac{1}{D_{\epsilon_{w_i}}\left(F_{ij}\right)}$.

\subsection{Analysis}
In this section, we analyze ICE-Buckets and bound both its maximum relative error and overall relative error.
\subsubsection{Maximum Relative Error}

Define the following function, $$m(\epsilon)=A_{\epsilon}(L-1)
=\frac{(1+2\epsilon^2)^{L-1}-1}{2\epsilon^2}(1+\epsilon^2),$$
which denotes the maximum representable value with error~$\epsilon$.

Define $\varepsilon(M)$ to be the smallest $\epsilon$ we need for an optimal estimation function with capacity at least $M$. For $M>L-1$, $\varepsilon(M)$ is the inverse of $m(\epsilon)$. For $M\le L-1$, $\varepsilon(M)$ is zero. $m(\epsilon)$ is increasing with $\epsilon$ and therefore $\varepsilon(M)$ is increasing with $M$.
We now show that ICE-Buckets' relative error is not larger than $\varepsilon(M)$.	

\begin{theorem} \label{ICE-Buckets maximum relative error}
The maximum relative error of ICE-Buckets can be bounded by $\varepsilon(M)$ for any distribution of the counters.
\end{theorem}
\begin{proof}
Construct an ICE-Buckets structure with $\epsilon_{step} = \frac{\varepsilon(M)}{E-1}$, where $E$ is the number of different possible estimation scales.
Let $M_i$ be the value of the biggest counter in bucket $i$. If we choose the scale of bucket $i$ to be at least
$w_i=\left\lceil
\frac{\varepsilon\left(M_i\right)}{\epsilon_{step}}\right\rceil$, we obtain an error parameter of no less than
$$\left\lceil\frac{\epsilon \left(M_i\right)}{\epsilon_{step}}\right\rceil \epsilon_{step}\ge\epsilon \left(M_i\right),$$
which gives us a capacity of at least $M_i$. In other words, we round up each bucket's error to the nearest product of~$\epsilon_{step}$.

$\varepsilon(M_i)\leq \varepsilon(M)$, because the maximum counter in each bucket is smaller or equal to the total number of packets and $\epsilon$ is increasing with $M$.
Thus, the maximum relative error over the entire counter scale is bounded by $\varepsilon(M)$.
\end{proof}

We conclude that ICE-Buckets has the same maximum relative error as the optimal estimation function.

\subsubsection{Overall Relative Error}

ICE-Buckets also improves the guaranteed overall relative error. In ICE-Buckets, this error is
\begin{equation}
\epsilon_{overall}=\sqrt{\frac{1}{N}\sum_{f}MSRE\left[n_f\right]}
\end{equation}

In Theorem~\ref{second bound}, we show an upper bound for $\epsilon_{overall}$ when  ICE-Buckets is configured optimally. In order to prove it, we need to show that $\epsilon^2(M)$ is concave for $M\ge L-1$ (as can be seen in Figure~\ref{epsilon square as M graph}). We observe that $m(\epsilon)$ is convex and increasing with $\epsilon^2$. The inverse of this function (which is defined on $M\ge L-1$) is therefore concave and increasing. Similarly, $\varepsilon(M)$ is also increasing and concave on $M\ge L-1$.

\begin{figure}[h]
\centering
\includegraphics[width=90mm]{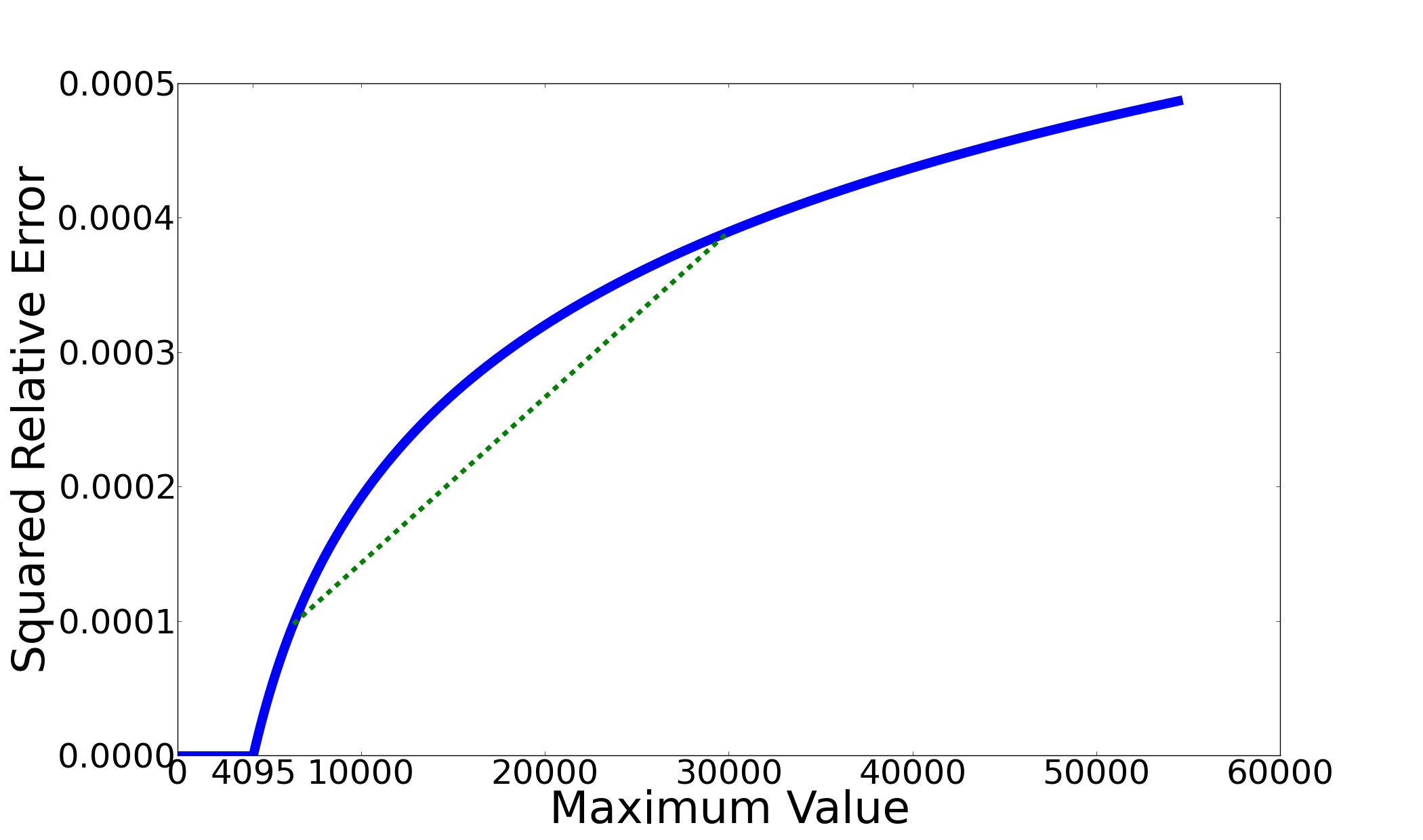}
\caption {$\epsilon^2$ as a function of $M$ for $L=4096$. The dashed line comes to show that the function is concave for $M\ge L-1$.}
\label{epsilon square as M graph}
\end{figure}
\begin{theorem} \label{second bound}
For any counter distribution, ICE-Buckets can be configured to have an overall relative error no greater than 
$$\epsilon \left( {M \over B}+L-1\right)+\frac{\varepsilon\left(M\right)}{E-1}.$$
\end{theorem}
\begin{proof}
Use the same construction as in the proof of Theorem~\ref{ICE-Buckets maximum relative error}, i.e., $\epsilon_{step} = \frac{\varepsilon(M)}{E-1}$ and $w_i=\left\lceil
\frac{\varepsilon\left(M_i\right)}{\epsilon_{step}}\right\rceil$. With this construction, the overall relative error is no greater than
$$\epsilon_{overall}\le
\sqrt{\frac{1}{N}\sum_{f}MSRE\left[n_f\right]}\le \sqrt{\frac{\sum_{i=0}^{B-1}\Big\lceil\frac{\epsilon \left(M_i\right)}{\epsilon_{step}}\Big\rceil^2 \epsilon_{step}^2}{B}}.$$
Instead of rounding $\varepsilon\left(M_i\right)$ to be a product of $\epsilon_{step}$, we can simply add $\epsilon_{step}$ and keep a bound that is no smaller.
\footnotesize
\begin{align*}
\epsilon_{overall}\le\sqrt{\frac{\sum\limits_{i=0}^{B-1}\Big\lceil\frac{\epsilon \left(M_i\right)}{\epsilon_{step}}\Big\rceil^2 \epsilon_{step}^2}{B}}\le
\sqrt {{{\sum\limits_{i = 0}^{B-1}
\left(\epsilon \left(M_i\right)+\epsilon_{step}\right)^2 } \over {B}}}&\\
\\
\end{align*}
\normalsize

We claim that according to the concaveness of $\epsilon$ and $\epsilon^2$,
$\left(\varepsilon(x)+\epsilon_{step}\right)^2$ is concave on $x>L-1$. This is because:
\begin{align*}
\footnotesize
\alpha&
\left(\varepsilon(x)+\epsilon_{step}\right)^2
+\left(1-\alpha\right)
\left(\varepsilon(y)+\epsilon_{step}\right)^2\\
=& \alpha\epsilon^2(x)
	 	   		+\left(1-\alpha\right)\epsilon^2(y)\\
& +2\epsilon_{step}\left(\alpha\varepsilon(x)+
				\left(1-\alpha\right)\varepsilon(y)\right)
				+\epsilon^2_{step}\\
\le& \epsilon^2\left(\alpha x 	+\left(1-\alpha\right)y\right)
    	 	+2\epsilon_{step}\varepsilon\left(
    		\alpha x+\left(1-\alpha\right)y\right)
    		+\epsilon^2_{step}\\
    	 \le& \left(\varepsilon\left(
    	 		\alpha x+\left(1-\alpha\right)y\right)
    	 		+\epsilon_{step}\right)^2.
\normalsize
\end{align*}

To use concaveness, we must make sure that all of the values are greater or equal to $L-1$. We do so by adding $L-1$ to each symbol. Since $\epsilon$ is an increasing function, the result is no smaller. That is,
$$\epsilon_{overall}\le \sqrt {{{\sum\limits_{i = 0}^{B-1}
\left(\epsilon \left(M_i+L-1\right)+\epsilon_{step}\right)^2 } \over {B}}}.$$
We can now apply Jensen's inequality:

\footnotesize
\begin{align*}
\epsilon_{overall}\le&\sqrt {{{\sum\limits_{i = 0}^{B-1}
	\left(\epsilon \left(M_i+L-1\right)+\epsilon_{step}\right)^2 } \over {B}}}\\
&\le\sqrt{\left(\epsilon \left(\frac{\sum\limits_{i = 0}^{B-1}
\left(M_i+L-1\right)}{B}\right)+\epsilon_{step}\right)^2}\\
&=\epsilon \left(\frac{M}{B}+L-1\right)+\epsilon_{step}.\\
\end{align*}
\normalsize

Setting $\epsilon_{step}=\frac{\varepsilon\left(M\right)}{E-1}$, we get an overall relative error of no more than
$$\epsilon_{overall}\le \epsilon \left( {M \over B}+L-1\right)+\frac{\varepsilon\left(M\right)}{E-1}.$$
\end{proof}
With the correct choice of parameters, this guaranteed overall relative error is far better than the one we can achieve with CEDAR - $\varepsilon(M)$.

This bound demonstrates the effect $E$ and $S$ have on the error. We want the bucket size $S=\frac{N}{B}$ to be as small as possible to restrict the negative impact on other counters, as a counter's size only affects the errors of counters sharing the same bucket.
As for $E$, we want it to be as big as possible to increase the error granularity.
However, decreasing $S$ or increasing $E$ increases the memory overhead.

Note that while CEDAR guarantees optimal estimation in terms of the maximum CV of the hitting time, there is no such claim in terms of the overall relative error.
This enables ICE-Buckets to significantly improve $\epsilon_{overall}$.

\subsection{Dynamic Configuration}
In Theorem \ref{second bound}, we have shown that there is an ICE-Buckets configuration for which the error is low.
We now show how to dynamically configure ICE-Buckets to fit any workload.
This process is composed of local and global upscale operations.

\subsubsection{Local Upscale}
The configuration of the data structure, $\{w_i\}_{i=0}^{B-1}$, is dynamically adjusted to the biggest estimation value in each bucket.  Initially, and bucket scales are set to Zero.
Whenever a symbol $F_{ij}$ approaches $L$, we increment $w_i$ and upscale bucket $i$ to use the parameter $\epsilon_{w_{i}+1}$. This is done by upscaling all of the flows in bucket $i$ using the symbol-upscale procedure described in Algorithm~\ref{symbolupscale}. A pseudo code of the local upscale procedure can be found in Algorithm~\ref{local upscale}.
We note that since the number of counters per bucket $S$ is small, local upscale can be efficiently implemented in hardware.

\begin{algorithm}
\caption{Local Upscale}\label{local upscale}
\begin{algorithmic}[1]
\scriptsize
\Procedure {UpscaleBucket}{i}

\For {$j=0$ to $S-1$}
	\State \Call{SymbolUpscale}{$F_{ij}$,
								$\epsilon_{w_i}$,
								$\epsilon_{w_i+1}$}
\EndFor
\State $w_i \gets w_i+1$
\EndProcedure
\normalsize
\end{algorithmic}
\end{algorithm}

\subsubsection{Global Upscale}
When a counter in a bucket with the maximum scale index ($E-1$) approaches its maximum value ($L-1$), we initiate a global upscale procedure  to prevent overflow. The procedure doubles the size of $\epsilon_{step}$.  Buckets with odd $w_i$s perform a local upscale. Then, every bucket $i$ updates its scale index to $w_i/2$. Pseudo code is given in Algorithm~\ref{bigupscale}.

\begin{algorithm}
\caption{Global Upscale}\label{bigupscale}
\begin{algorithmic}[1]
\scriptsize
\Procedure{GlobalUpscale}{}

\For {$u=0$ to $B-1$}
	\If {$w_u\mod2=1$}
		\State \Call{UpscaleBucket}{$u$}
	\EndIf
	\State $w_u\gets \frac{w_u}{2}$
\EndFor
\State $\epsilon_{step}\gets 2\epsilon_{step}$
\EndProcedure
\normalsize
\end{algorithmic}
\end{algorithm}

Table~\ref{tbl:upscale} demonstrates this process. In this case, we use eight different possible scales ($E=8$) and  $\epsilon_{step}=0.1\%$ before the global upscale. Buckets with odd $w_i$ perform local upscale and their scale is incremented accordingly, so that all buckets have even scale parameters and we can safely half their scales to match the new $\epsilon_{step}$. Notice that at the moment of upscale, $\epsilon_{max}$ is increased only by $\epsilon_{step}$. In this case, no bucket after upscale has a scale parameter larger than~4. In general, no bucket has a scale parameter larger than~$\frac{E}{2}$ immediately after global upscale. 
\begin{table}[htp]
{
\begin{center}
\footnotesize
\begin{tabular}{|c|c|c|c|c|}

\hline
Old $w$ & Old $\epsilon_{w}$  & New $w$ & New $\epsilon_w$ & Requires Upscale?\tabularnewline
\hline
\hline
0 & 0.0\% & 0 & 0.0\% & No\tabularnewline
\hline
1 & 0.1\% & 1 & 0.2\% & Yes\tabularnewline
\hline
2 & 0.2\% & 1 & 0.2\% & No\tabularnewline
\hline
3 & 0.3\% & 2 & 0.4\% & Yes\tabularnewline
\hline
4 & 0.4\% & 2 & 0.4\% & No\tabularnewline
\hline
5 & 0.5\% & 3 & 0.6\% & Yes\tabularnewline
\hline
6 & 0.6\% & 3 & 0.6\% & No\tabularnewline
\hline
7 & 0.7\% & 4 & 0.8\% & Yes\tabularnewline
\hline
--- & --- & 5 & 1.0\% & ---\tabularnewline
\hline
--- & --- & 6 & 1.2\% & ---\tabularnewline
\hline
--- & --- & 7 & 1.4\% & ---\tabularnewline
\hline
\end{tabular}
\normalsize
\end{center}
{\caption{Global upscale example ($E = 8$); $\epsilon_{step}$ is updated from 0.1\% to 0.2\%.} \label{tbl:upscale}}}
\end{table}

Global upscale may be difficult to implement in hardware.
A method to upscale the entire counter array while continuously counting new packet arrivals is described in \cite{CEDAR}.
This method also applies to ICE-Buckets' global upscale.
In our case, global upscale can be completely avoided by setting the maximal error to $\varepsilon(M)$.

\subsection{Parameter Choice}
We now describe the process of choosing the parameters to minimize the upper bound from Theorem~\ref{ICE-Buckets maximum relative error} and then Theorem~\ref{second bound}. In the standard scenario, we have limited space for our data structure of $T$ bits. We usually have an upper bound for $M$, e.g., by multiplying the maximum supported traffic rate by the maximum measurement time. If $M$ is still unknown, we can use the maximum integer we can represent. $N$ could also be given, as the maximum number of flows the networking device supports. If $N$ grows during run-time and we have enough space, we can always allocate more counters on the fly. We now choose $L$, $B$, $S$ and $E$. To minimize $\epsilon_{max}$, we should allocate as many bits as possible for every counter. We therefore allocate $\log_2 L = \lfloor\frac{T}{N}\rfloor$ bits per counter. We are left with $T\mod N$ bits for the scale parameters (if no memory is left we can use a single bucket). Next, we note that there is no point in choosing $E$ to be larger than $M$. Every upscale should increase the counting capacity by at-least one and therefore $M$ upscales should be always sufficient to achieve the maximum counting capacity. To find the optimal $E$, we can iterate over $\log_2 E$, which should be an integer number as it represents the number of bits we give the scale parameter. For each choice of $E$, we calculate the number of buckets we can afford: $B=\lfloor\frac{T\mod N}{\log_2 E}\rfloor$.  Given all the parameters, we can calculate the upper bound from Theorem~\ref{second bound}. The upper bound requires the computation of the function $\varepsilon(M)$. This function can be computed through binary search because the opposite function is increasing and can be easily computed. By iterating over the possible $E$ values, we can find the $E$ that gives the smallest upper bound. After finding $E$, we can calculate $B$ and $S=\frac{N}{B}$.

\begin{figure}[htp]
\centering
\includegraphics[width=\columnwidth]{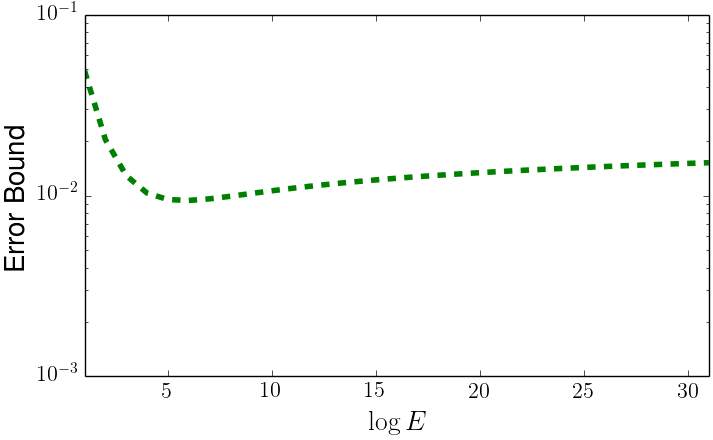}
\caption{The error bound from Theorem~\ref{second bound} for $N,M$ of NZ09, different choices of $E$ and  an average of 12.5 bits per counter.}
\label{fig:parameter choice}
\end{figure}

For example, consider the trace NZ09, which will be presented in Section~\ref{Simulation Results}. The trace has $N=32,737,760$ flows. In this example, we allocate $12.5$ bits for each counter. $M$ is unknown in advance so we choose the maximum int $2^{32}-1$. After allocating the maximum of $12$ bits per symbol, we are left with $T \mod N=16,368,880$ bits. We now try all possible options of $E$, and for each option we calculate $B$ and the upper bound from Theorem~\ref{second bound}. Figure~\ref{fig:parameter choice} depicts the computed bounds for every choice of $E$. We can see that for $E$s that are too small, the granularity of the error $\epsilon_{step}$ is too crude and the result is a high error. The optimal $E$ for the upper bound is $2^6$, and larger $E$s give higher errors because they require allocation of bigger buckets.

\begin{figure*}[htp]
	\subfigure[CHI08 with 12-bit symbols]{\includegraphics[width=\columnwidth]
		{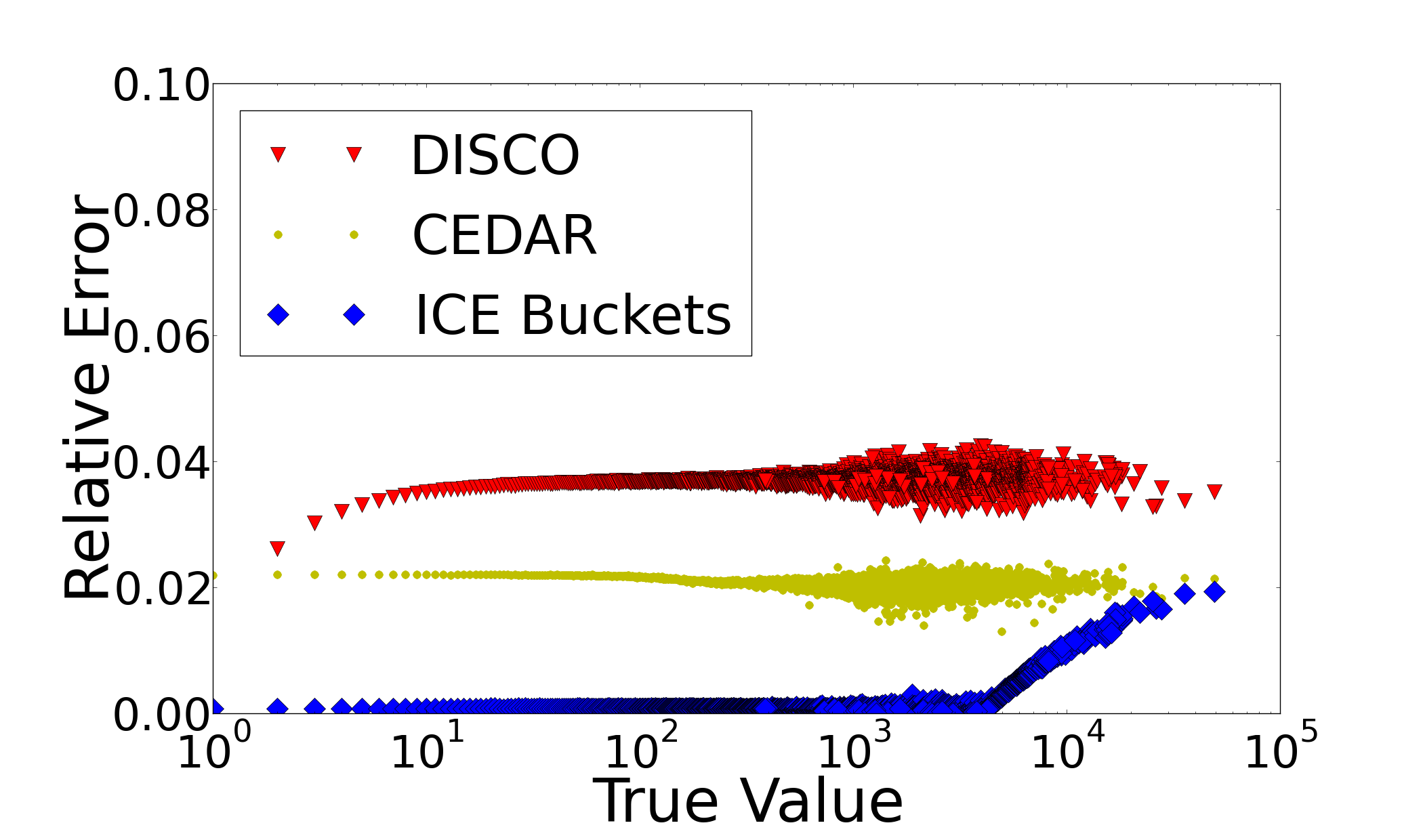}}
	\subfigure[CHI08 with 8-bit symbols ]{\includegraphics[width=\columnwidth]
			{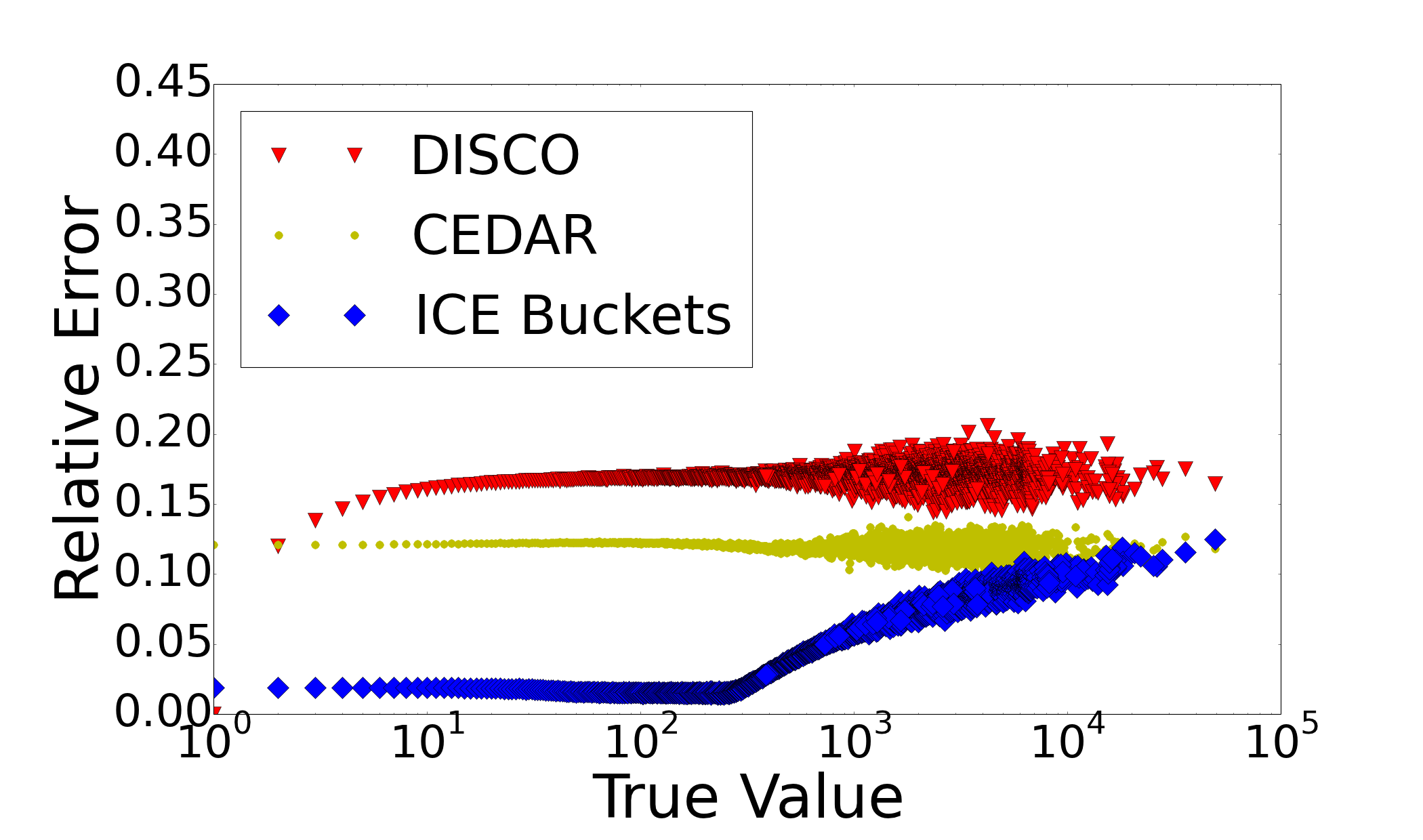}}
	\subfigure[NZ09 with 12-bit symbols]{\includegraphics[width=\columnwidth]
		{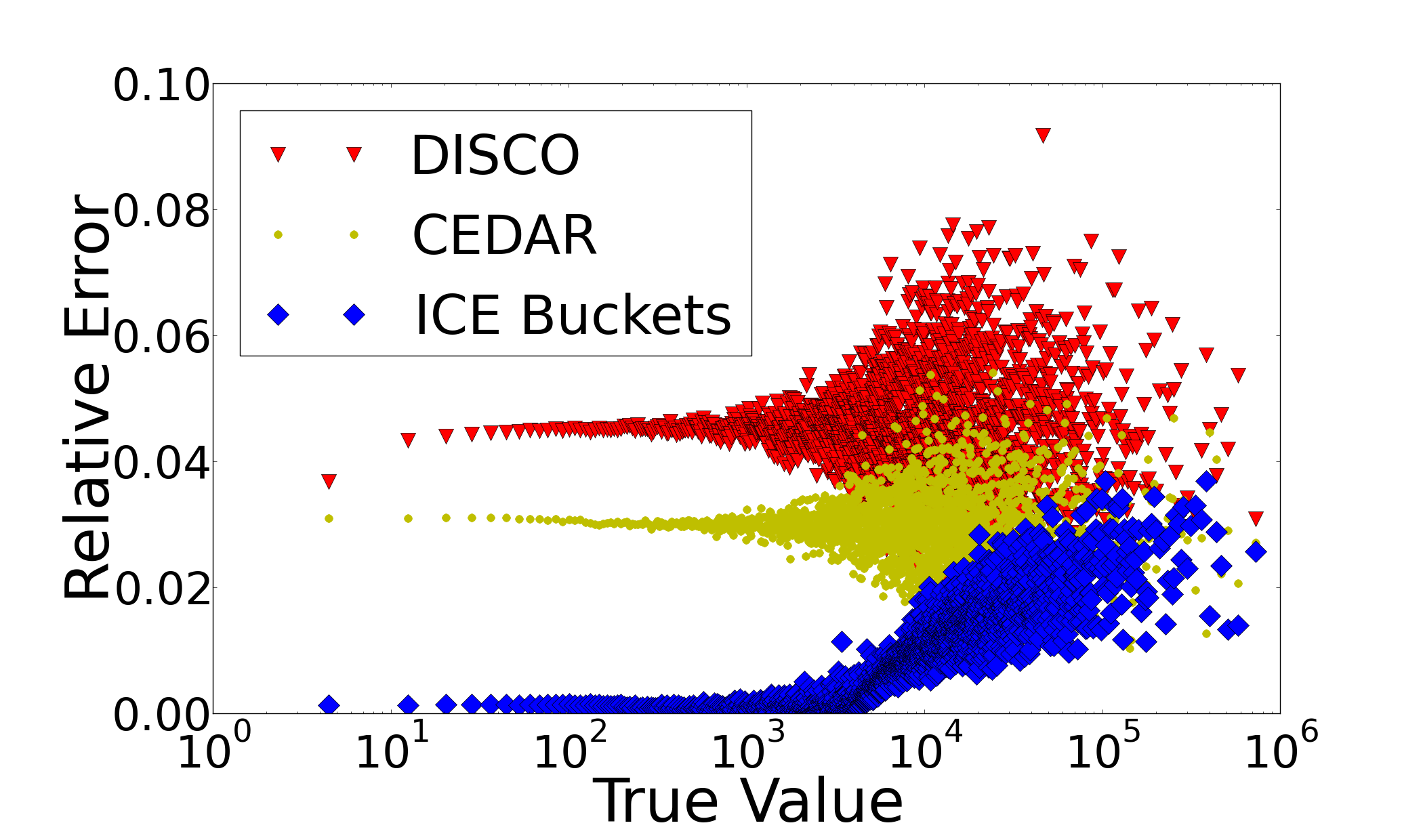}}
	\subfigure[NZ09 with 8-bit symbols]{\includegraphics[width=\columnwidth]
		{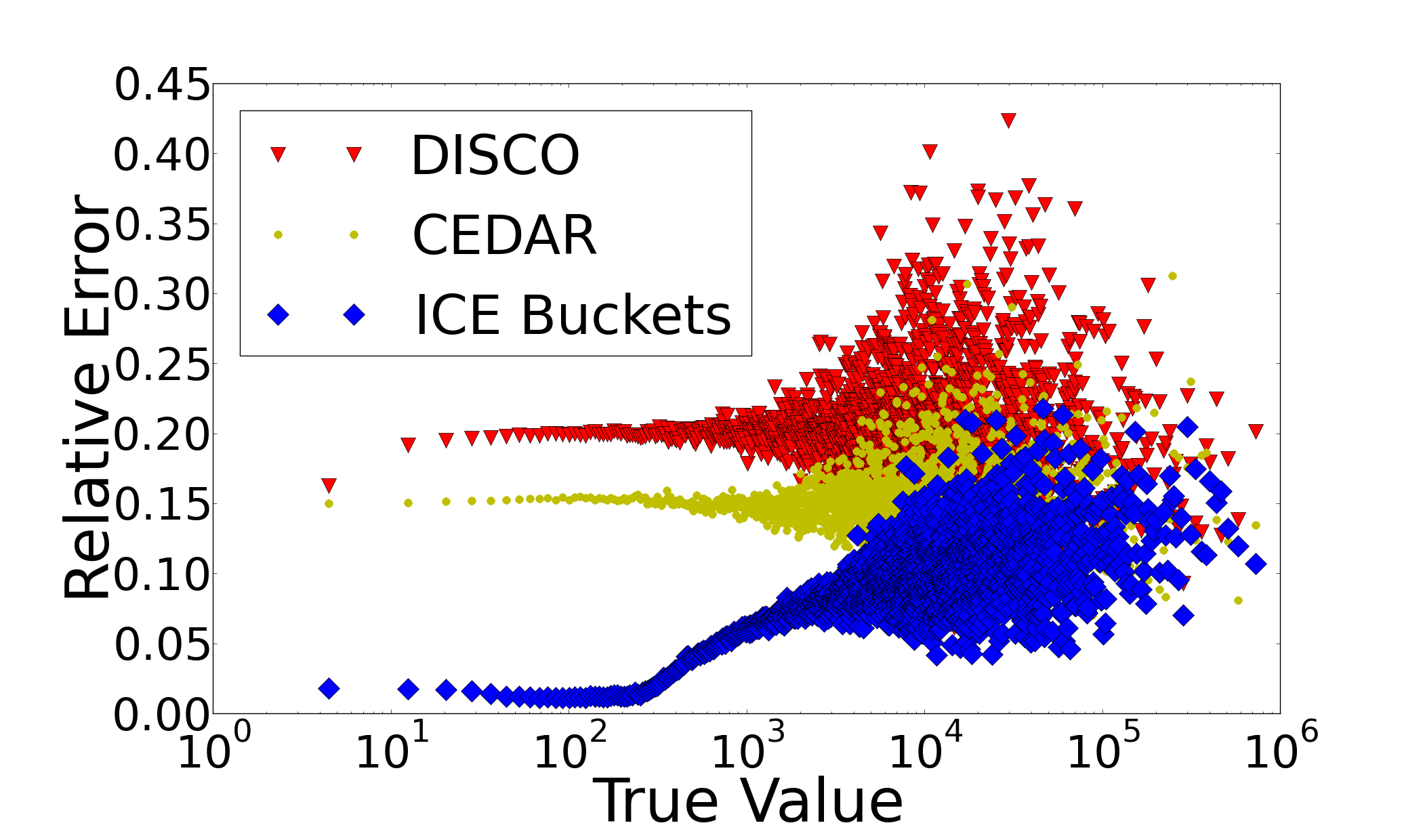}}
	\subfigure[CHI15 with 12-bit symbols]{\includegraphics[width=\columnwidth]
		{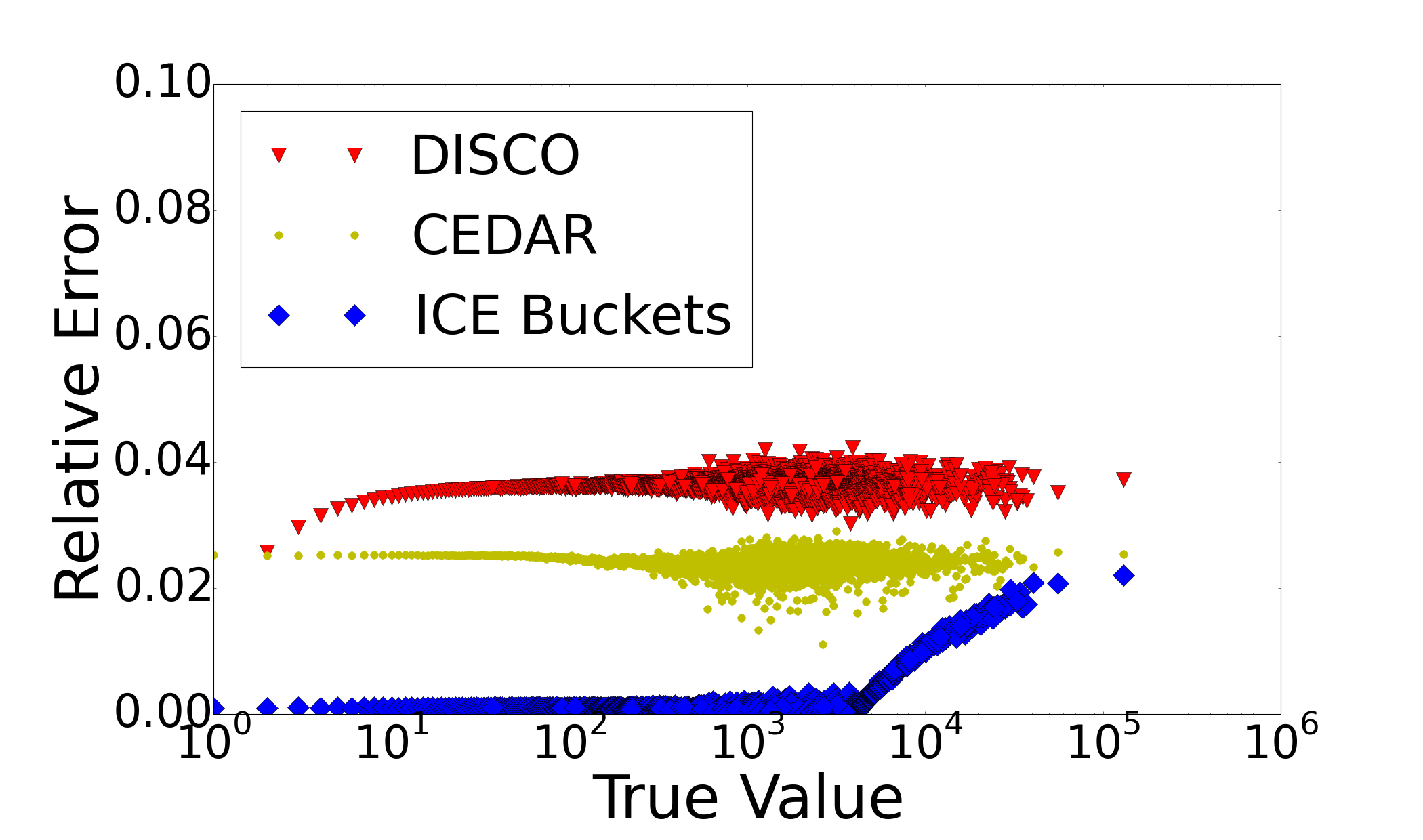}}
	\subfigure[CHI15 with 8-bit symbols]{\includegraphics[width=\columnwidth]
		{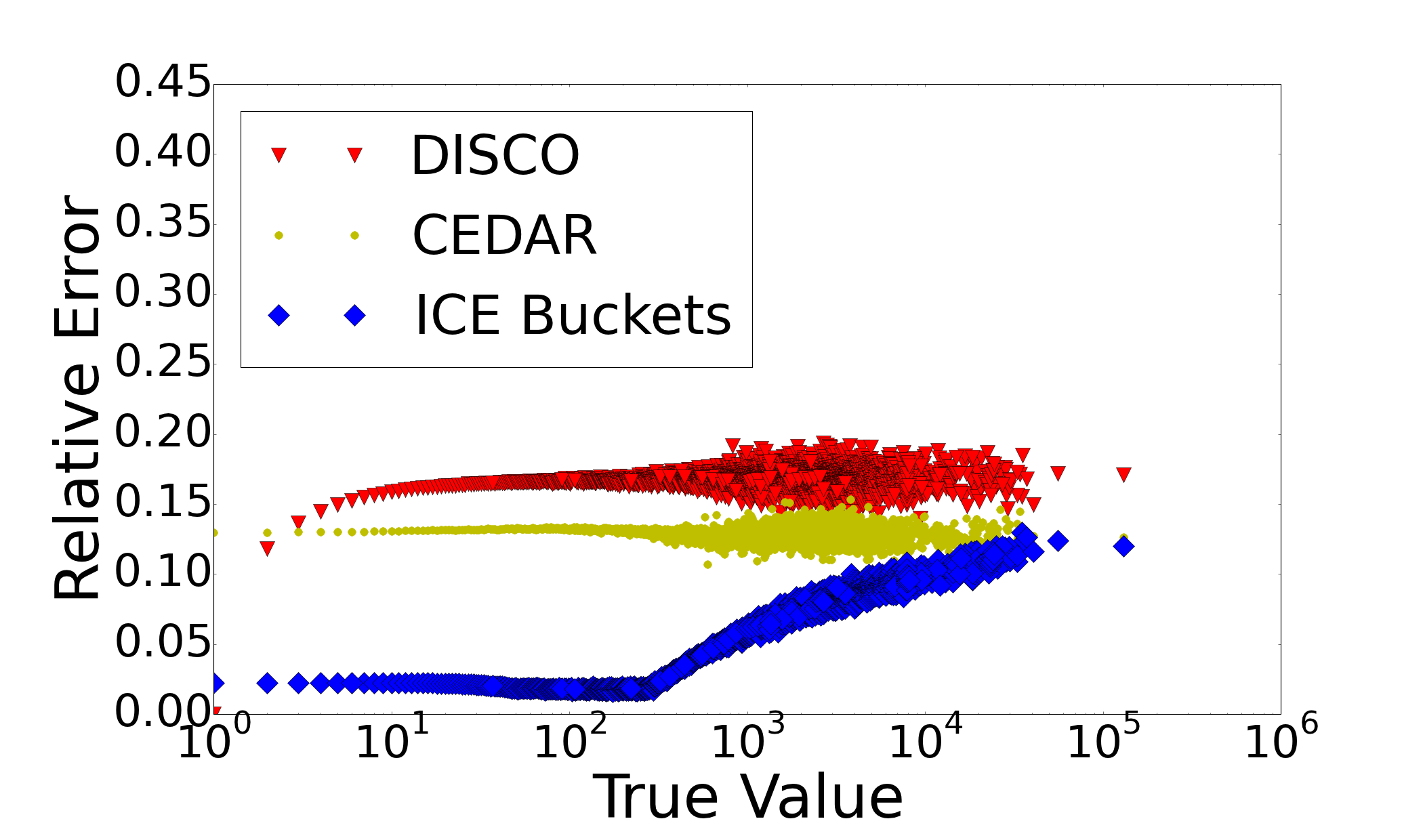}}
	\caption{Comparison of the relative error per value on different traces using 8 and 12-bit symbols}
	\label{errorPerValue}
\end{figure*}

\section{Simulation Results}
\label{Simulation Results}

We evaluated ICE-Buckets with five different Internet packet traces. 
The first trace (NZ09) consists of twenty four hours of Internet traffic collected from an unnamed New Zealand ISP on Jan. 6th 2009~\cite{NetworkingTrace2}.
It is a relatively large trace, containing almost a billion packets and over 32 Million flows.
We also used data from two Equinix data-centers in the USA, which are connected to backbone links of Tier1 ISPs. Both use per flow data for load balancing. One is connected to a link between Chicago, IL and Seattle, WA. We used three traces from this dataset, CHI08 from 2008~\cite{NetworkingTrace1}, CHI14 from 2014~\cite{NetworkingTrace4} and CHI15 from 2015~\cite{NetworkingTrace5}. CHI08 was previously used to evaluate CEDAR in~\cite{CEDAR}.
Trace (SJ13) was taken in 2013 from an Internet data collection monitor that is connected to a link between San Jose and Los Angeles, CA~\cite{NetworkingTrace3}.

We compare ICE-Buckets to two state of the art counter estimation algorithms - DISCO~\cite{DISCOJournal} and CEDAR~\cite{CEDAR}.
In order to measure different memory constraints, we tested each algorithm with both 8.5 and 12.5 bits on average per counter.
We use 8 and 12 bits (correspondingly) for the symbols.
The remainder is used for the scale parameters in ICE-Buckets and for storing the estimation value array in CEDAR.
Previous works were evaluated with similar symbol lengths.
DISCO does not have an upscaling scheme.
Therefore we configured it according to the maximal expected number of packets ($M$) that is different for each trace, as specified in Table~\ref{parameters}.
Per-trace statistics and configurations are given in Table~\ref{parameters}.

\footnotesize
\begin{table*}[htb]
{
\begin{center}
\begin{tabular}{|c|c|c|c|c|c|}
\hline
Trace & Flows (\textit{N}) & Packets & \textit{M} & \textit{S}      & \textit{E} \tabularnewline
      &           &         &   & 8b / 12b & 8b / 12b
\tabularnewline
\hline
\hline
CHI08 & 1,420,318 & 26,750,712 & 26,750,712 & 10 / 14 & 32 / 128\tabularnewline
\hline
CHI14 & 1,213,614 & 34,721,808 & 34,721,808 & 10 / 12 & 32 / 64\tabularnewline
\hline
SJ13 & 3,071,187 & 20,803,060 & 20,803,060 & 12 / 14 & 64 / 128\tabularnewline
\hline
NZ09 & 32,737,760 & 891,023,765 & $2^{32}-1$ & 12 / 12 & 64 / 64\tabularnewline
\hline
CHI15 & 683,708 & 18,774,214 & 18,774,214 & 10 / 12 & 32 / 64\tabularnewline
\hline
\end{tabular}
\end{center}
}

{\caption{Documentation of trace characteristics and ICE-Buckets configurations used in experiments for 8/12 bit symbols.} \label{parameters}}
\end{table*}

\begin{table*}[tb]
{
\begin{center}
\begin{tabular}{|c|c|c|c|c|c|c|c|c|c|c|}
\hline
Trace & \multicolumn{2}{c|}{CHI08} & \multicolumn{2}{c|}{CHI14} & \multicolumn{2}{c|}{SJ13} & \multicolumn{2}{c|}{NZ09} & \multicolumn{2}{c|}{CHI15}\tabularnewline
Bits-Per-Symbol & \multicolumn{1}{c}{8} & 12 & \multicolumn{1}{c}{8} & 12 & \multicolumn{1}{c}{8} & 12 & \multicolumn{1}{c}{8} & 12 & \multicolumn{1}{c}{8} & 12\tabularnewline
\hline
\hline
CEDAR (upper bound) & 16.93\% & 3.70\% & 17.10\% & 3.75\% & 16.76\% & 3.65\% & 19.61\% & 4.51\% & 16.70\% & 3.64\%\tabularnewline
\hline
ICE-Buckets (upper bound) & 5.02\% & 0.42\% & 5.71\% & 0.50\% & 3.49\% & 0.26\% & 8.22\% & 0.94\% & 5.63\% & 0.49\%\tabularnewline
\hline
\hline
DISCO (actual) & 10.34\% & 2.24\% & 10.16\% & 2.24\% & 7.82\% & 1.71\% & 14.24\% & 3.21\% & 11.58\% & 2.51\%\tabularnewline
\hline
CEDAR (actual) & 12.19\% & 2.17\% & 12.40\% & 2.38\% & 13.05\% & 2.60\% & 15.01\% & 3.09\% & 12.96\% & 2.51\%\tabularnewline
\hline
ICE-Buckets (no global upscale) & 1.70\% & 0.10\% & 2.14\% & 0.14\% & 1.00\% & 0.03\% & 1.78\% & 0.14\% & 4.29\% & 0.21\% \tabularnewline
\hline
ICE-Buckets (actual) & 1.50\% & 0.06\% & 1.96\% & 0.16\% & 1.01\% & 0.03\% & 1.81\% & 0.13\% & 2.02\% & 0.11\%\tabularnewline
\hline
\end{tabular}
\end{center}
}

{\caption{Overall relative error bounds and actual overall relative error of different algorithms on various traces}
\label{comparison}}
\end{table*}

\normalsize

Table~\ref{comparison} presents the overall relative error of ICE-Buckets and the alternatives for the tested traces.
We also present the upper bounds on the error of ICE-Buckets and CEDAR. As can be observed, for real datasets, ICE-Buckets' error is much lower than this bound, since the majority of flows are small. For CHI08, ICE-Buckets achieves an overall relative error that is over 57 times smaller than that of CEDAR.
Notice that for all traces, ICE-Buckets' overall relative error with 8-bit symbols is lower than that of the alternatives, even with 12-bit symbols.
Note that in our case, DISCO is $100\%$ accurate when the value is 1, and is slightly less accurate than CEDAR for all other values.
All in all, this results in an overall relative error similar to that of CEDAR.

We also experimented with a version of ICE-Buckets that does not use global upscale, which should be easier to implement in hardware.
To do so, we pre-configured $\epsilon_{step}$ to ensure local upscales are sufficient to count $2^{32}$-1 packets in any bucket.
Note that for most traces the error in this case is very similar to that of ICE-Buckets with upscale. We therefore recommend to implement ICE-Buckets without global upscale when the total number of packets can be bounded in advance.

To explain the cause of ICE-Buckets' substantial error reduction, we show in
Figure~\ref{errorPerValue} the relative error as a function of the real counter value. The relative error was computed from 256 runs of each algorithm on CHI08 and CHI15 and one run on NZ09. Note that for most counter values, the relative error of ICE-Buckets is lowest, followed by CEDAR, and then DISCO. ICE-Buckets achieves an error close to zero for counters smaller than $L$ because an accurate counter of $log_2L$~bits suffices to represent those values. Unfortunately, this error cannot always be zero, as some of these counters share buckets with larger counters. As the counter scale grows, the estimation error increases, until eventually, the largest counter is estimated with $\epsilon_{max}$. In contrast, CEDAR estimates all of the counters with relative error $\epsilon_{max}$.

\begin{figure*}[htp]
	\subfigure[CHI08]{\includegraphics[width=90mm]{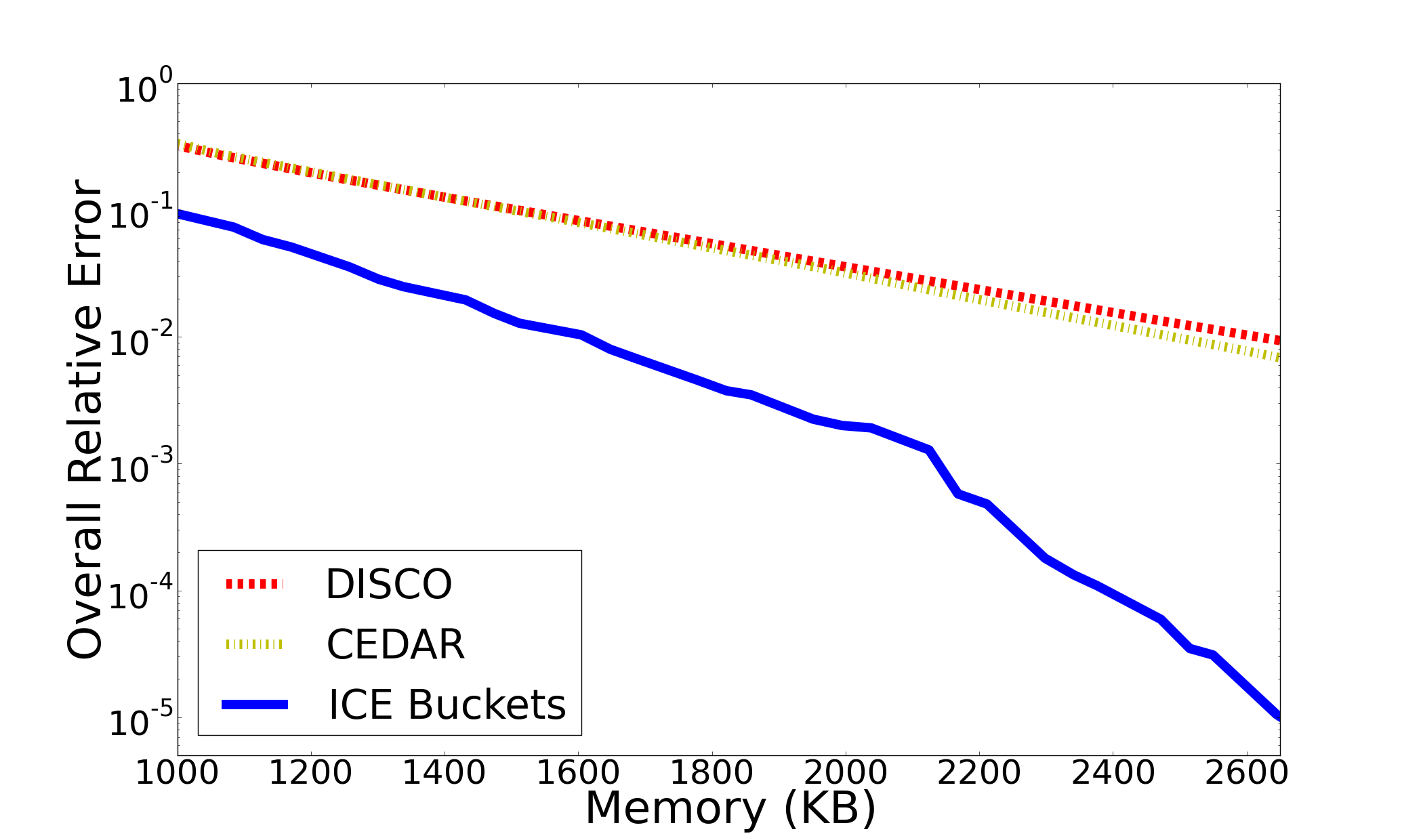}}
	\subfigure[CHI15]{\includegraphics[width=90mm]{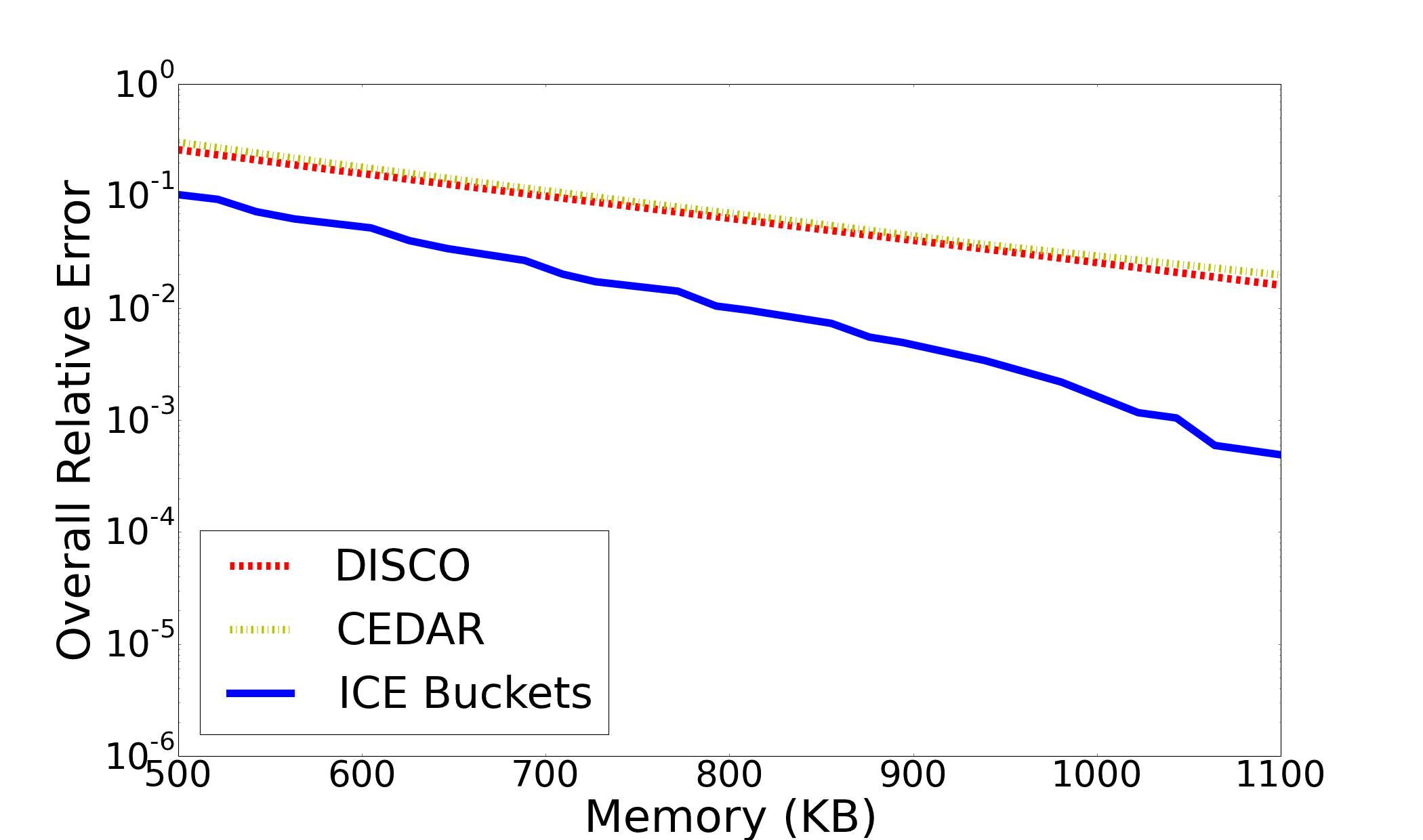}}
\caption{A comparative evaluation of the accuracy as a function of the allocated memory space; overheads are taken into account!}
\label{error_per_memory}
\end{figure*}

Figure~\ref{error_per_memory} illustrates the accuracy of different algorithms for CHI08 and CHI15 under varying memory constraints. Overheads of all methods are taken into account and the maximal symbol size is used for each method. ICE-Buckets uses different configurations with overheads that range between $1 \over 4$ and $3 \over 4$ bits per counter. The figure shows that DISCO and CEDAR provide similar space-accuracy trade-offs as mentioned in Section~\ref{sec:relative error}. Under all of the simulated memory constraints, ICE-Buckets is more accurate than both CEDAR and DISCO, and the difference between ICE-Buckets and the other algorithms grows with the memory. We explain this by noting that as the number of bits per symbol (L) grows, more counters can be estimated with zero error.

Figure~\ref{error_per_time} describes the overall relative error of CEDAR and ICE-Buckets throughout the NZ09 trace's progress. To adapt to the growing counter scale, both ICE-Buckets and CEDAR use an upscale mechanism that gradually increases the error. Note that the overall relative error of ICE-Buckets is almost constant throughout an entire day of real Internet traffic. In addition, CEDAR's multiple global upscales are clearly visible in the figure. In contrast, ICE-Buckets' upscales are mostly local and cause a smoother increase in the relative error.

\section{Conclusion}
\label{Conclusion}

In this work, we have introduced ICE-Buckets, a novel counter estimation data structure that minimizes the relative error.
ICE-Buckets uses the optimal estimation function with a scale that is optimized independently for each bucket.

We first described an explicit representation of this function, which was previously known only in recursive form. We extended its analysis and showed a method to measure the effect of upscale operations on the relative error. This function is used in ICE-Buckets to minimize the error in each bucket.

ICE-Buckets is dynamically configured to adapt to the growing counters. For practical deployments, it can be implemented without global operations while providing similar accuracy. 

We proved an upper bound to ICE-Buckets' overall relative error, which is significantly smaller than that of previous estimation algorithms. In particular, we demonstrated a reduction of up to 14 times in this upper bound when applied to traffic characteristics of real workloads. ICE-Buckets also achieves the same maximum relative error as the optimal function.

Additionally, we extensively evaluated ICE-Buckets with four Internet packet traces
and demonstrated a reduction of up to 57 times in overall error. ICE-Buckets achieves an improvement in accuracy even when it is given considerably less space than the alternatives.
Finally, we have shown that ICE-Buckets is significantly more accurate than the leading alternatives for a wide range of memory constraints.

In this work, we explained how to perform decrements and downscaling.
Yet, doing so, greatly complicates the analysis.
Analyzing their impact is left for future work.

As mentioned before, another interesting topic for future work is combining shared counters schemes like CMS~\cite{CountMinSketch}, multi-stage filters~\cite{HuffmanBF}, SBF~\cite{SpectralBloom}, as well as TinyTable~\cite{TinyTable} with estimators.
Since ICE-buckets offers small counters with low error, replacing the counters in the above with estimators could potentially improve their space to accuracy ratio.
Another promising aspect of the above is that estimators can return the estimated value in $O(1)$ time, thereby maintaining the access efficiency of such combined schemes.

\begin{figure}[htp]
	\centering
	\includegraphics[width=90mm]{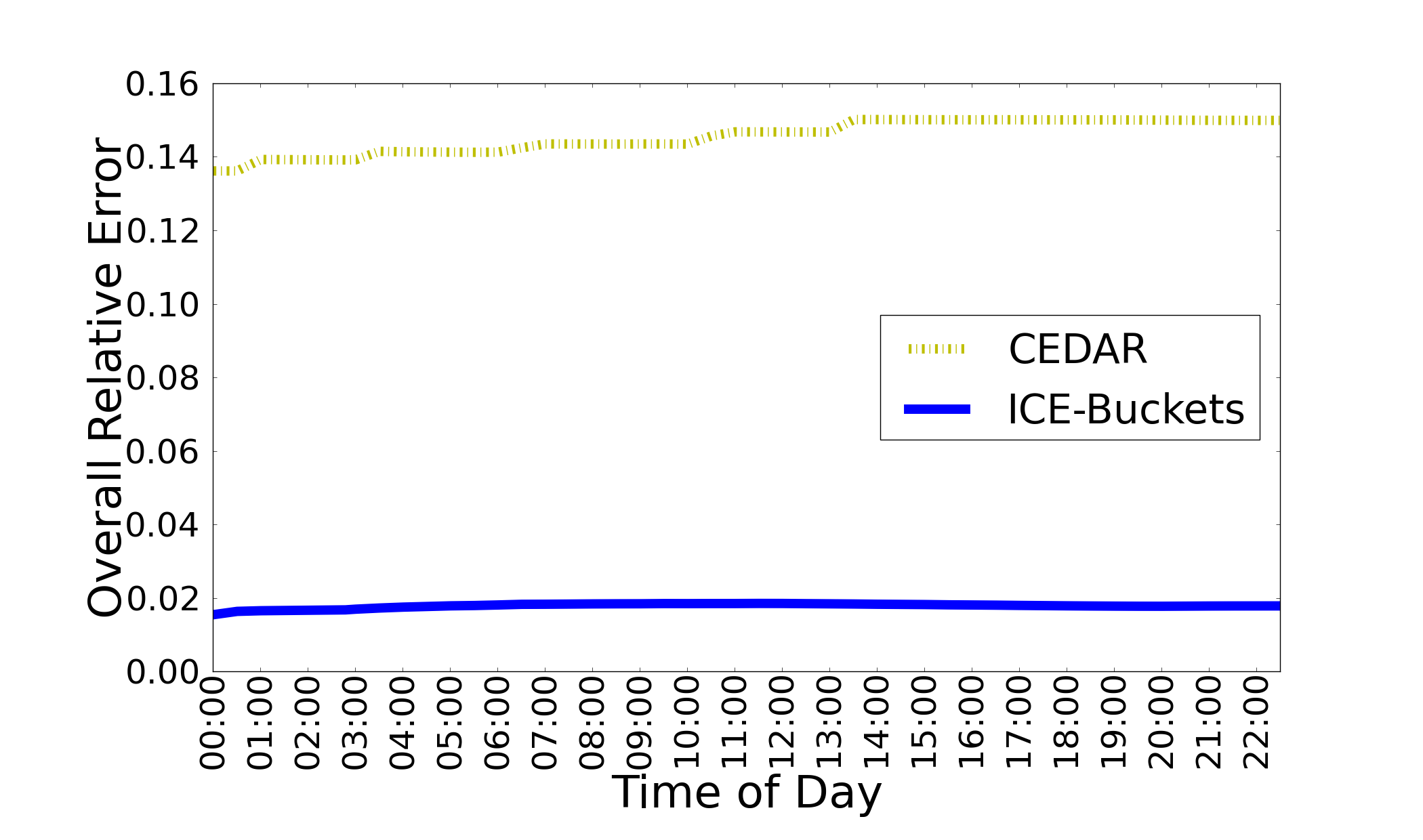}
	\caption{Average relative error through trace progress with 8-bit symbols as simulated on NZ09}
	\label{error_per_time}
\end{figure}

\paragraph*{Acknowledgements} We would like to thank Isaac Keslassy, Olivier Marin and Erez Tsidon for their insights and helpful advice. 
This work was partially funded by the Israeli Ministry of Science and Technology grant 3-10886 and the Technion HPI center.
{
	\bibliographystyle{IEEEtran}
	\bibliography{refs}
}
\end{document}